\documentclass[twocolumn]{autart}    % Enable this line and disable the 
\usepackage{graphicx}          % Include this line if your 
\usepackage{color}
\usepackage{caption}
\usepackage{natbib}        % required for bibliography
\usepackage{subfigure}
\usepackage{float}
\usepackage{algorithm}
\usepackage{amssymb}     
\usepackage{amsmath}
\usepackage{bm}
\usepackage{booktabs} 
\usepackage{etoolbox}
\usepackage{array}
\usepackage{color}
\usepackage{multirow}
\usepackage{makecell}
\usepackage{enumerate}

\let\classAND\AND
\let\AND\relax
\usepackage{algorithmic}
\let\AND\classAND
	\AtBeginEnvironment{algorithmic}{\let\AND\algoAND}

	\newenvironment{proof}{\textbf{Proof:}}{\hfill  $\blacksquare$}

	\usepackage{array}

\newtheorem{theorem}{Theorem}
\newtheorem{lemma}{Lemma}
\newtheorem{definition}{Definition}
\newtheorem{remark}{Remark}
\newtheorem{proposition}{Proposition}
\newtheorem{assumption}{Assumption}

\newtheorem{example}{Example}
\newtheorem{update}{Update Rule}
\newtheorem{problem}{Problem}
	
\begin{document}
	
\begin{frontmatter}
	%\runtitle{Insert a suggested running title}  % Running title for regular 
	% papers but only if the title  
	% is over 5 words. Running title 
	% is not shown in output.
	
	\title{Dynamic Leader-Follower Consensus with Adversaries: \\A Multi-Hop Relay Approach\thanksref{footnoteinfo}} % Title, preferably not more 
	% than 10 words.
	
	\thanks[footnoteinfo]{This work was supported in part by the National Natural Science
		Foundation of China under Grant 62403188 and in part by JSPS
		under Grants-in-Aid for Scientific Research Grant No. 22H01508 and
		24K00844. The material in this paper was not presented at any conference.}

	\author[Changsha]{Liwei Yuan}\ead{yuanliwei@hnu.edu.cn}, 
	\author[Tokyo]{Hideaki Ishii}\ead{hideaki\_ishii@ipc.i.u-tokyo.ac.jp}
	
	\address[Changsha]{College of Electrical and Information Engineering, Hunan University, Changsha 410082, China}
	\address[Tokyo]{Department of Information Physics and Computing, The University of Tokyo, Tokyo 113-8656, Japan}

	                             % keyword list or with the 
	% help of the Automatica 
	% keyword wizard

	\begin{abstract}                          % Abstract of not more than 200 words.
		This paper examines resilient dynamic leader-follower consensus within multi-agent systems, where agents share first-order or second-order dynamics. The aim is to develop distributed protocols enabling nonfaulty/normal followers to accurately track a dynamic/time-varying reference value of the leader while they may receive misinformation from adversarial neighbors. Our methodologies employ the mean subsequence reduced algorithm with agents engaging with neighbors using multi-hop communication. We accordingly derive a necessary and sufficient graph condition for our algorithms to succeed; also, our tracking error bounds are smaller than that of the existing method. Furthermore, it is emphasized that even when agents do not use relays, our condition is tighter than the sufficient conditions in the literature. With multi-hop relays, we can further obtain more relaxed graph requirements. Finally, we present numerical examples to verify the effectiveness of our algorithms.
	\end{abstract}
	
\end{frontmatter}

\section{Introduction}

Amid growing concerns over cyber security in multi-agent systems (MASs), consensus protocols in the presence of adversarial agents creating failures and attacks have garnered significant interests, e.g.,
\cite{teixeira2012attack}.
Within this area, resilient consensus problems have gained substantial attention across the disciplines of systems control, distributed computing, and robotics (\cite{vaidya2012iterative,sundaram2018distributed,yu2022resilient}). Here, the objective for the nonfaulty, normal agents is to reach consensus despite misbehaviors of adversarial agents. 
Existing resilient consensus algorithms are designed to ensure that normal agents reach consensus on a value within the convex hull of their initial states, e.g.,
\cite{yuan2021secure,yuan2023event,yu2022resilient}. Meanwhile, numerous formation control and reliable broadcast problems require agents to reach consensus on a predetermined reference value, which may lie inside or outside that convex hull (\cite{bullo2009distributed}). Furthermore, in certain formation control applications where the reference value is time-varying (\cite{cao2009distributed,zegers2022event}), achieving consensus on this dynamic value over time is particularly advantageous; such applications include distributed target tracking performed by multiple unmanned aerial vehicles (UAVs) in outdoor environments (\cite{hassanalian2017classifications}). Therefore, this motivates our development of resilient algorithms for these requirements.

Previous works have examined fault-free leader-follower consensus, where followers agree on the reference value of a leader (\cite{dimarogonas2009leader}). Moreover, for MASs with adversaries, multiple works have investigated related problems of reliable broadcast (\cite{koo2006reliable}) and the certified propagation algorithm (CPA) (\cite{koo2004broadcast, tseng2015broadcast}); they aim for a secure leader to broadcast a reference value to all normal nodes in the network. 
Additionally, several works have investigated a related problem called resilient distributed estimation (e.g., \cite{leblanc2014resilient}). Relevantly, \cite{mitra2019byzantine} explored a problem in which observation information of a system is resiliently transmitted from a group of source nodes to other nodes that cannot directly observe the system. Overall, these reliable broadcast and distributed estimation methods are not applicable to tracking arbitrary dynamic values.

In this paper, we study resilient dynamic leader-follower consensus in a directed network with adversaries, where normal followers track the time-varying reference value of the leader.
Extensive works have utilized the so-called mean subsequence reduced (MSR) algorithms to tackle leaderless resilient consensus and have established tight conditions on graph structures for MSR algorithms to succeed, e.g., \cite{vaidya2012iterative,leblanc2013resilient}. Later, \cite{usevitch2020resilient} have proved a sufficient condition for the sliding weighted-MSR (SW-MSR) algorithm (\cite{saldana2017resilient}) to achieve resilient leader-follower consensus to arbitrary static reference values.
Subsequently, in our recent work (\cite{yuan2024reaching}), we proposed a necessary and sufficient condition for the same problem.
Note that it could be difficult to apply these methods for tracking a time-varying reference due to large tracking errors in this situation.
On the other hand, \cite{rezaee2021resiliency} studied resilient leader-follower consensus to dynamic reference values in directed networks with relatively small tracking errors. Moreover, \cite{zegers2022event} proposed a method for normal followers to detect adversarial neighbors and track a time-varying reference value of the trustworthy leader, but it is limited to undirected networks and requires an upper bound on the leader's value.
Recently, for leaderless resilient consensus, the introduction of multi-hop relays has relaxed the stringent graph connectivity requirements in \cite{khan2019exact,yuan2023event,yuan2021resilient}; it enables messages of an agent to be relayed to further destinations by intermediate agents.
%Inspired by these works, our aim is to relax the graph connectivity requirement for our leader-follower problem through multi-hop relays.

The contributions of this paper are summarized as follows. 
First, we propose two novel algorithms based on the multi-hop weighted MSR (MW-MSR) algorithm (\cite{yuan2021resilient}) to tackle resilient dynamic leader-follower consensus in first-order and second-order MASs with directed topologies, respectively. We characterize a necessary and sufficient graph condition for our algorithms to succeed under adversarial Byzantine agents.
It is notable that even with one-hop communication, our condition is tighter than the ones in the related works with static reference values (\cite{usevitch2020resilient}) and dynamic reference values (\cite{rezaee2021resiliency}).
Besides, by introducing multi-hop relays, our method can increase the network robustness against adversaries without changing the topology. As a result, our approach can tolerate more adversarial nodes compared to the one-hop MSR algorithms (\cite{usevitch2020resilient,rezaee2021resiliency}) and the CPA works (\cite{koo2004broadcast,tseng2015broadcast}).

We emphasize that the extension to second-order MASs is crucial since for agents in robotics, second-order dynamics can describe their motions more precisely; see, e.g., \cite{paranjape2018robotic}. To the best of our knowledge, this case has not been investigated in the literature. 
It is worth noting that compared to \cite{rezaee2021resiliency}, our approaches have smaller consensus error bounds.
Besides, we allow the leader's value to be arbitrary as long as its velocity is bounded.
We also provide an analysis for resilient dynamic leader-follower consensus with insecure leaders, which is closely related to our main results. 
Lastly, numerical examples show our algorithms' potential for strengthening the security of formation control.

The rest of this paper is organized as follows. 
Section~2 outlines the problem settings. 
Section~3 defines our graph notion.
Section~4 analyzes the convergence of our method.
Section~5 presents another algorithm for second-order MASs.
Section~6 presents numerical examples to verify the efficacy of our algorithms.
Section~7 concludes the paper.

\section{Preliminaries and Problem Settings}

\subsection{Graph Notions}
Consider a digraph $\mathcal{G} = (\mathcal{V},\mathcal{E})$ consisting of the node set $\mathcal{V}=\{1,...,n-1,d\}$ and the edge set $\mathcal{E}$. The edge $(j,i)\in \mathcal{E}$ indicates that node $i$ can get information from node $j$.
The subgraph of $\mathcal{G}$ induced by the node set $\mathcal{H}\subset\mathcal{V}$ is $\mathcal{G}_\mathcal{H}=(\mathcal{H},\mathcal{E}(\mathcal{H}))$, where $\mathcal{E}(\mathcal{H})=\{(i,j)\in \mathcal{E}: i,j\in \mathcal{H}\}$.
An $l$-hop path ($l\in \mathbb{Z}_{>0}$) from source node $i_1$ to destination node $i_{l+1}$ is a sequence of distinct nodes $(i_1, i_2, \dots, i_{l+1})$, where $(i_j, i_{j+1})\in \mathcal{E} $ for $j=1, \dots, l$.
Let $\mathcal{N}_i^{l-}$ be the set of nodes that can reach node $i$ via paths of at most $l$ hops.
Let $\mathcal{N}_i^{l+}$ be the set of nodes that are reachable from node $i$ via paths of at most $l$ hops. Node $i$ is included in both sets above.
For the $l$-th power of $\mathcal{G}$, denoted by $\mathcal{G}^l$, it contains the node set $\mathcal{V}$ and its adjacency matrix $A = [a_{ij} ]$ is given by $\alpha \leq a_{ij}<1$ if $j\in \mathcal{N}_i^{l-}$ and $a_{ij} = 0$ otherwise, where $\alpha > 0$ is a constant. Denote by $|\mathcal{S}|$ the cardinality of a finite set $\mathcal{S}$.

In this paper, $\textup{sgn}(\cdot)$ denotes the sign function. For a fixed parameter $\epsilon \in \mathbb{R}_{>0}$, $\textup{sat}_\epsilon(\cdot)$ is the saturation function defined for a scalar $x$:
\begin{equation}\label{sat}
	\textup{sat}_\epsilon(x) = \left\{
	\begin{aligned} 
		1,&  & x > \epsilon,   \\
		x/\epsilon,&  & -\epsilon \leq  x \leq \epsilon,   \\
		-1,&   &x < -\epsilon.  
	\end{aligned}
	\right.
\end{equation} 
Next, we describe our communication model. 
%Node $i_1$ sends its messages to an $l$-hop neighbor $i_{l+1}$ via different paths at time $k$.
We represent a message as a tuple $m=(w,P)$, where $w=\mathrm{value}(m)\in \mathbb{R}$ is the message content and $P=\mathrm{path}(m)$ indicates the path via which $m$ is transmitted. 
At time $k\geq 0$, each normal node $i$ exchanges the messages $m_{ij}[k]=(x_i[k],P_{ij}[k])$ consisting of its state $x_i[k]\in \mathbb{R}$ along each path $P_{ij}[k]$ with its multi-hop neighbor $j$ via the relaying process in \cite{yuan2021resilient}.
Denote by $\mathcal{V}(P)$ the set of nodes in $P$.

\subsection{System Model and Resilient Algorithm}\label{problemsetting}
Consider the leader-follower MAS modeled by the digraph $\mathcal{G} = (\mathcal{V},\mathcal{E})$, where $\mathcal{V}$ consists of one leader agent $d$ and the set of follower agents $\mathcal{W}$ with $d\notin \mathcal{W}$ and $\{d\}\cup \mathcal{W}=\mathcal{V}$. 
The leader $d$ propagates a dynamic reference value to followers in $\mathcal{W}$, and thereafter, followers achieve consensus on that reference value. However, if adversarial agents are present, they may prevent normal followers from achieving the objective.

In the literature (\cite{usevitch2020resilient,rezaee2021resiliency}), leader agents have been categorized into four cases depending on if they are secure or not, and if they are known to followers or not. In this paper, we focus on the case where leader agents are secure and known to direct neighbor followers. In this case, if all secure leaders broadcast the same reference value, they can be viewed as one secure leader;\footnote{When secure leaders broadcast different reference values, leader-follower consensus cannot be achieved for the whole MAS including all leaders as leaders usually are not affected by others during the consensus process; see also \cite{ren2007multi,usevitch2020resilient,rezaee2021resiliency}. } thus, in Assumption~\ref{secured}, we study a secure leader that is known to followers as \cite{rezaee2021resiliency} did.
Moreover, we will extend our analysis for the other three cases in Section~\ref{discussion_notsecured}.

Next, we describe our system under attacks. Denote the set of adversarial agents by $\mathcal{A}\subset \mathcal{W}$ and denote the set of normal agents by $\mathcal{N}=\mathcal{V}\setminus\mathcal{A}$. The formal definition of adversarial agents is given later. Then, the set of normal follower agents is denoted by $\mathcal{W}^\mathcal{N}=\mathcal{W}\cap \mathcal{N}$. Besides, $\{d\}\cup \mathcal{W}^\mathcal{N}=\mathcal{N}$ holds under Assumption~\ref{secured}. Denote the set of direct followers by $\mathcal{W}_d = \{ i\in \mathcal{W}: (d,i)\in \mathcal{E} \}$.

\begin{assumption}\label{secured}
	The leader agent $d$ is secure, i.e., $d\in \mathcal{N}$. Moreover, the direct followers in $ \mathcal{W}_d$ know the node index of the secure leader agent $d$.
\end{assumption}

At each time $k$, the leader $d$ updates its state $x_d[k] \in \mathbb{R}$ according to its dynamics with bounded inputs and propagates $x_d[k]$ to followers. The leader's dynamics will be given in Section~\ref{sec_first} for the first-order type and Section~\ref{sec_second} for the second-order type. We assume that the leader's value changing rate (or speed) is bounded as follows. Let $T$ be the sampling period of the system.
\begin{assumption}\label{leader-speed}
	The value changing rate $v_d[k]$ of leader $d$ is bounded, i.e., $|v_d[k]|=|x_d[k+1]-x_d[k]|/T \leq \overline{v}_d , \forall k\geq 0$, where  $\overline{v}_d$ is a positive constant and is known to all followers in $\mathcal{W}$.
\end{assumption}

Then, we define our first resilient dynamic leader-follower consensus problem in this paper, which is also studied in \cite{rezaee2021resiliency}.
\begin{problem}\label{problem}
	Design a distributed control strategy such that the normal agents in $\mathcal{N}$ reach resilient dynamic leader-follower consensus, i.e.,
	for any possible sets and behaviors of the adversaries in $\mathcal{A}$ and any state values of the normal agents in $\mathcal{N}$, the following condition is satisfied for a given error bound $\overline{\epsilon}>0$:
	\begin{equation}\label{reach_consensus}
		\exists \medspace \overline{k} < \infty \medspace\medspace \textup{s.t.} \medspace\medspace   |x_i[k]-x_d[k]|\leq \overline{\epsilon},  \forall i\in  \mathcal{N}, \forall k \geq \overline{k}.
	\end{equation} 
\end{problem}

\vspace{-2mm}
Next, we introduce the multi-hop weighted MSR (MW-MSR) algorithm from \cite{yuan2021resilient} in Algorithm~1, which is employed in our resilient dynamic leader-follower consensus protocols.
The notion of minimum message cover (MMC) is crucial in Algorithm~1.

\begin{definition} For a graph $\mathcal{G} = (\mathcal{V},\mathcal{E})$, let $\mathcal{M}$ be a set of messages transmitted through $\mathcal{G}$, and let $\mathcal{P}(\mathcal{M})$ be the set of message paths of all the messages in $\mathcal{M}$, i.e., $\mathcal{P}(\mathcal{M}) =\{\mathrm{path}(m):m \in \mathcal{M}\}$. A \textit{message cover} of $\mathcal{M}$ is a set of nodes $\mathcal{T}(\mathcal{M})\subset \mathcal{V}$ whose removal disconnects all message paths, i.e., for each path $P\in \mathcal{P}(\mathcal{M})$, we have $\mathcal{V}(P)\cap \mathcal{T}(\mathcal{M})\neq \emptyset$. In particular, an MMC of $\mathcal{M}$ is defined by
	\begin{equation*}
		\mathcal{T}^*(\mathcal{M})\in	\arg \min_{\substack{ \mathcal{T}(\mathcal{M}): \textup{ Cover of } \mathcal{M}}} 	\left|  \mathcal{T} (\mathcal{M})\right| . 
	\end{equation*}
\end{definition}

\vspace{-2mm}
In Algorithm~1, normal node $i$ can remove the largest and smallest values from exactly $f$ nodes located no more than $l$ hops away. With multi-hop relays, it requires the MMC to determine the number of such extreme values. See \cite{yuan2021resilient} for more details of Algorithm~1.
As we will present later, Algorithm~1 is able to enhance the network robustness using small relay hops $l$, where one can calculate the MMC in polynomial time. Hence, compared to the one-hop algorithms (\cite{usevitch2020resilient,rezaee2021resiliency}), Algorithm~1 is slightly more computationally complex.

\begin{algorithm}[t] %\footnotesize
	\caption{MW-MSR Algorithm} 
	\begin{algorithmic}
		\REQUIRE Node $i$ knows $x_i[0]$, $\mathcal{N}_i^{l-}$, $\mathcal{N}_i^{l+}$.
		
		\STATE At each time $k$, for $\forall i \in \mathcal{N}$:
		
		\STATE \textbf{1)} Exchange messages:
		
		\STATE Send $m_{ij}[k]=(x_i[k],P_{ij}[k])$ to $\forall j\in \mathcal{N}_i^{l+}$. 
		
		\STATE Receive $m_{ji}[k]=(x_j[k],P_{ji}[k])$ from $\forall j\in \mathcal{N}_i^{l-}$ and store them in the set $\mathcal{M}_i[k]$.
		
		\STATE Sort $\mathcal{M}_i[k]$ in an increasing order based on the message values (i.e., $x_j[k]$ in $m_{ji}[k]$).
		
		\STATE \textbf{2)} Remove extreme values:

		\STATE (a) Define two subsets of $\mathcal{M}_i[k]$:
		\vspace{1mm}
		
		\STATE \hspace{2.2mm} $\overline{\mathcal{M}}_i[k]=\{ m\in \mathcal{M}_i[k]: \mathrm{value}(m)> x_i[k]  \}$,
		\vspace{1mm}
		
		\STATE \hspace{2.2mm} $\underline{\mathcal{M}}_i[k]=\{ m\in \mathcal{M}_i[k]: \mathrm{value}(m)< x_i[k]  \}$.
		\vspace{1mm}
		
		\STATE (b) Get $\overline{\mathcal{R}}_i[k]$ from $\overline{\mathcal{M}}_i[k]$:

		\IF{$\left|  \mathcal{T}^* (\overline{\mathcal{M}}_i[k])\right| <f$}
		\STATE $\overline{\mathcal{R}}_i[k] = \overline{\mathcal{M}}_i[k]$;
		\ELSE
		\STATE Choose $\overline{\mathcal{R}}_i[k]$ s.t. (i)
		$\forall m\in \overline{\mathcal{M}}_i[k]\setminus \overline{\mathcal{R}}_i[k]$, $\forall m'\in \overline{\mathcal{R}}_i[k]$, $\mathrm{value}(m) \leq \mathrm{value}(m')  \medspace\medspace \medspace \medspace  $ 
		and (ii) $\left|  \mathcal{T}^* (\overline{\mathcal{R}}_i[k])\right| =f$. 
		\ENDIF
		
		\STATE (c) Get $\underline{\mathcal{R}}_i[k]$ from $\underline{\mathcal{M}}_i[k]$ similarly, which contains smallest message values.
		
		\STATE (d) $\mathcal{R}_i[k]=\overline{\mathcal{R}}_i[k]\cup\underline{\mathcal{R}}_i[k]$.

		\STATE \textbf{3)} Update:
		\vspace{-4mm}
					\begin{align}\label{msrupdate}
							\phi_i[k] = &\sum_{m \in \mathcal{M}_i[k]  \setminus \mathcal{R}_i[k]}   \big(x_i[k]   - \mathrm{value}(m )  \big).
					\end{align}
		\vspace{-5mm}
		\ENSURE $\phi_i[k]$.
	\end{algorithmic}
\end{algorithm}

\subsection{Adversary Model}

We introduce our adversary models (\cite{vaidya2012iterative,leblanc2013resilient,yuan2021resilient}).
\begin{definition}
	\textit{($f$-local set)}
	The adversary set $\mathcal{A}$ is said to be $f$-local (in $l$-hop neighbors)
	if $\left|\mathcal{N}_i^{l-} \cap \mathcal{A}\right| \leq f, \forall i\in \mathcal{N}$.
\end{definition}

\begin{definition}
	\textit{(Byzantine nodes)}
	An adversary node $i\in \mathcal{A}$ is said to be Byzantine
	if it arbitrarily changes its state/relayed values and sends different state/relayed values to its neighbors at each time.
\end{definition}

Byzantine agents are often used to characterize misbehaviors of adversarial agents in point-to-point networks (\cite{Lynch,vaidya2012iterative}). In contrast, the so-called \textit{malicious} agents are limited to send the same false information to neighbors and they form a suitable model for broadcast networks (\cite{goldsmith2005wireless}).

\begin{assumption}\label{assum_topology}
	Each normal node $i\in \mathcal{N}$ knows the $f$-local adversary model and the topology information of its neighbors up to $l$ hops, i.e., the paths from/to any neighbor $j\in \mathcal{N}_i^{l-} \cup \mathcal{N}_i^{l+}$.
\end{assumption}

	As is standard in \cite{Lynch, su2017reaching, yuan2021resilient}, we assume that normal nodes know the upper limit on the number of adversaries and neighbors' topology information.
	Note that we consider finite (typically small) relay hops $l$, and thus, Assumption~\ref{assum_topology} is more relaxed than the ones requiring each normal node to know the entire network topology in, e.g., \cite{tseng2015broadcast,khan2019exact}.

Lastly, we introduce Assumption~\ref{assumption_path} from \cite{su2017reaching,yuan2021resilient} merely for ease of analysis. In fact, manipulating message paths can be easily detected and such detection can be handled parallel to the case of manipulating message values. Related discussions have been presented in \cite{yuan2021resilient}.

\begin{assumption}\label{assumption_path}
	Each node $i\in \mathcal{A}$ can manipulate its state $x_i[k]$ and the values in messages that they send or relay, but cannot change the path $P$ in such messages. 
\end{assumption}

\section{Robust Following Graphs}\label{sec_robustness}

In this section, we present the notion of robust following graphs enabling our algorithms to achieve resilient dynamic leader-follower consensus in directed networks.

	We start with the definition of $r$-reachable followers; it characterizes the local graph structure for a node $i\in \mathcal{S}$ to be affected by the normal nodes outside $\mathcal{S}$ when node $i$ applies Algorithm~1.

\begin{definition}\label{reachability}
	\textit{($r$-reachable followers)}
	Consider the graph $\mathcal{G} = (\mathcal{V},\mathcal{E})$ with $l$-hop communication. For $r\in \mathbb{Z}_{>0}$ and a nonempty set $\mathcal{S}\subset \mathcal{V}$, we say that a node $i\in \mathcal{S}$ is an $r$-reachable follower with $l$ hops if it holds that
	\begin{equation*}
		|\mathcal{I}_{i, \mathcal{S}}^l| \geq r, 
	\end{equation*}
	where $\mathcal{I}_{i, \mathcal{S}}^l$ is the set of independent paths\footnote{Note that only node $i$ is common in these paths.} to node $i$ of at most $l$ hops originating from nodes outside $\mathcal{S}$.
\end{definition}

Then, we are ready to define $r$-robust following graphs with $l$ hops as follows. Since normal direct followers know the node index of the secure leader, we need to focus only on the graph structure excluding the leader.

\begin{definition}\label{robust_following}
	\textit{($r$-robust following graphs)}
	Consider the graph $\mathcal{G} = (\mathcal{V},\mathcal{E})$ with a secure leader $d\in \mathcal{V}$ and a set of followers $\mathcal{W}= \mathcal{V}\setminus \{d\}$. Let set $\mathcal{F} \subset \mathcal{W}$ and denote by $\mathcal{G}_{\mathcal{H}}$ the subgraph of $\mathcal{G}$ induced by node set $\mathcal{H}=\mathcal{W}\setminus\mathcal{F}$.
	$\mathcal{G}$ is said to be an $r$-robust following graph with $l$ hops (under the $f$-local model) if for any $f$-local set $\mathcal{F}$, the subgraph $\mathcal{G}_{\mathcal{H}}$ satisfies that for every nonempty subset $\mathcal{S}\subseteq \mathcal{H}\setminus \mathcal{W}_d$, the following condition holds:
	\begin{align}\label{zrs}
		| \mathcal{Z}_{\mathcal{S}}^{r}| \geq 1  ,
	\end{align}
	where $ \mathcal{Z}_{\mathcal{S}}^{r}
	= \{i\in \mathcal{S} :  |\mathcal{I}_{i, \mathcal{S}}^l| \geq r \}$.
	Moreover, if $\mathcal{W}_d=\mathcal{W}$, we also say $\mathcal{G}$ is an $r$-robust following graph with $l$ hops, where $r\leq |\mathcal{W}|=n-1$. 
\end{definition}
%\vspace{2mm}

We emphasize that for Definition~\ref{robust_following}, the graph robustness increases as the relay range $l$ increases; this is because $|\mathcal{I}_{i, \mathcal{S}}^l|$ increases as $l$ grows.
Next, we illustrate this idea using the graphs in Figs.~\ref{9node} and \ref{15node}. 

\begin{example}\label{discussion9node}
	Consider the graph in Fig.~\ref{9node}(a).
	It is not a $2$-robust following graph with $1$ hop under the $1$-local model. For example, after removing the node set $\mathcal{F}=\{5\}$, in the subgraph $\mathcal{G}_{\mathcal{H}}$, the set $\mathcal{S}=\{1,2,3,6\}$ does not satisfy \eqref{zrs}; i.e., $\mathcal{Z}_{\mathcal{S}}^{2}=\emptyset$ when $l=1$. 
	In fact, for this network to be a $2$-robust following graph with $1$ hop, four edges should be added as shown in Fig.~\ref{9node}(b). Alternatively, we could increase the network robustness by increasing the relay range. For example, when $l=2$, for node sets $\mathcal{F}=\{5\}$ and $\mathcal{S}=\{1,2,3,6\}$, node 2 has 2 independent paths originating from nodes outside $\mathcal{S}$, i.e., \eqref{zrs} is met.
	Lastly, this graph can be verified to be a $2$-robust following graph with $2$ hops under the $1$-local model; it needs to check all the combinations of node subsets.
\end{example}

\begin{figure}[t]
	\centering
	\subfigure[\scriptsize{ }]{
		\includegraphics[width=0.29\linewidth ]{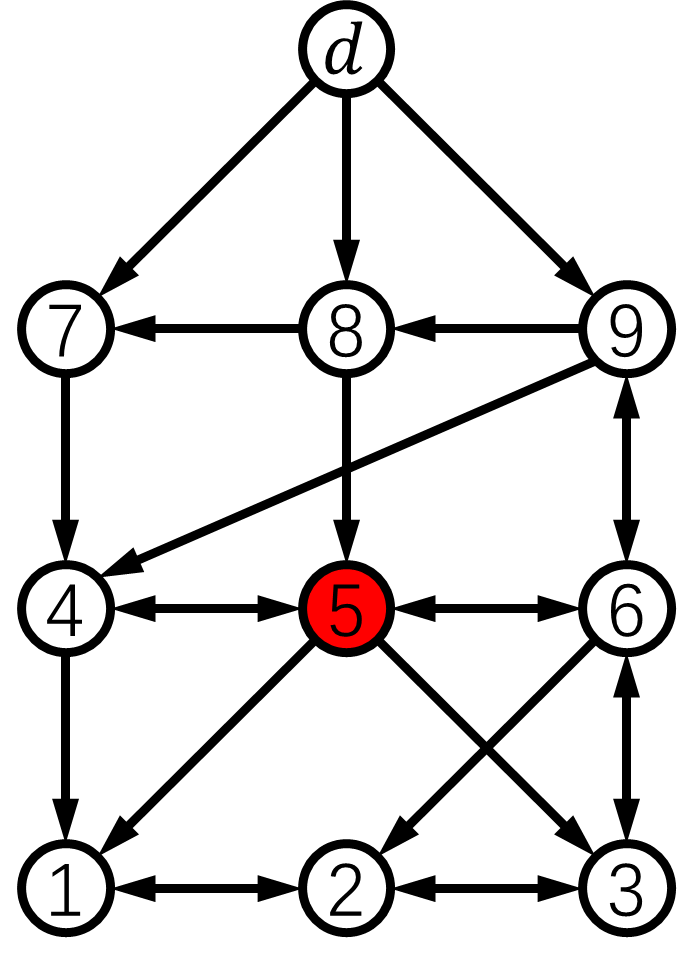}
	}
	\quad\quad
	\subfigure[\scriptsize{ }]{
		\includegraphics[width=0.29\linewidth ]{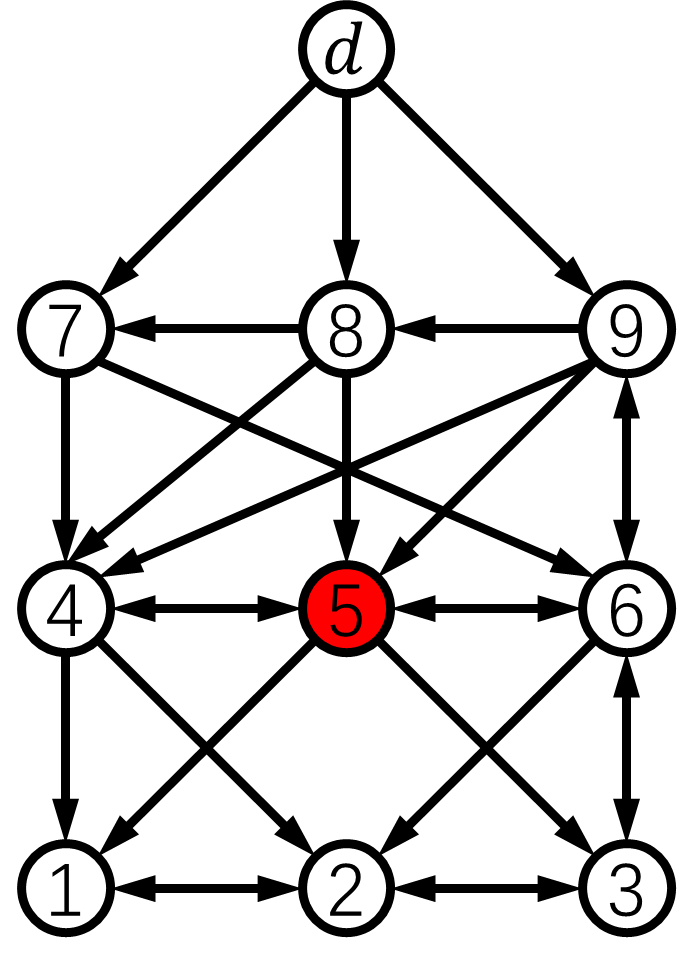}
	}
	\vspace{-9pt}
	\caption{(a) The graph is not a 2-robust following graph with 1 hop but is a 2-robust following graph with 2 hops under the 1-local model. (b) The graph is a 2-robust following graph with 1 hop under the 1-local model.}
	\label{9node}
	%\vspace*{-1.0mm}
\end{figure}

\begin{example}\label{discussion15node}
	Consider the larger graph in Fig.~\ref{15node}.
	It is not a $3$-robust following graph with $1$ hop under the $2$-local model.
	Notice that by selecting the set $\mathcal{F}=\{8,9\}$, in the subgraph $\mathcal{G}_{\mathcal{H}}$ with $\mathcal{H}= \mathcal{W}\setminus\mathcal{F}$, the set $\mathcal{S}=\mathcal{H}\setminus \mathcal{W}_d$ does not meet \eqref{zrs} when $l=1$ since none of the nodes has 3 in-neighbors outside $\mathcal{S}$.
	Meanwhile, when $l=3$, for the same sets $\mathcal{F}$ and $\mathcal{S}$, \eqref{zrs} is satisfied as node 2 has 3 independent paths originating from nodes outside $\mathcal{S}$.
	By checking all the combinations, we can conclude that it is a $3$-robust following graph with $3$ hops under the $2$-local model.
\end{example}

\begin{remark}\label{remark_leaderless}
	We compare the notion of robust following graphs with that of robustness with $l$ hops, which is the basis of tight graph conditions guaranteeing leaderless resilient consensus under the malicious model (\cite{leblanc2013resilient,yuan2021resilient}) and the Byzantine model (\cite{yuan2023event}). The notion of $r$-robustness with $l$ hops requires that for any two nonempty disjoint node sets $\mathcal{V}_1, \mathcal{V}_2 \subset \mathcal{V}$, at least one set includes an $r$-reachable node (similar to Definition~\ref{reachability}). 
	In contrast, robust following graphs are defined using one set $\mathcal{S}\subseteq \mathcal{H}\setminus \mathcal{W}_d$. This distinction arises from the nature of the two problems. Specifically, leaderless consensus aims at reaching consensus on a value that is not predetermined. Thus, the normal nodes in $\mathcal{V}_1$ (or $\mathcal{V}_2$) either influence or are influenced by those outside the set. Hence, two node sets are necessary for characterizing such potentially bidirectional information flows. 
	However, the leader-follower case requires followers to track the leader's value. Consequently, it is necessary that the followers in each set $\mathcal{S}$ are influenced by the normal nodes outside $\mathcal{S}$ via sufficient independent paths.
\end{remark}

\begin{figure}[t]
	\centering
	\includegraphics[width=0.56\linewidth ]{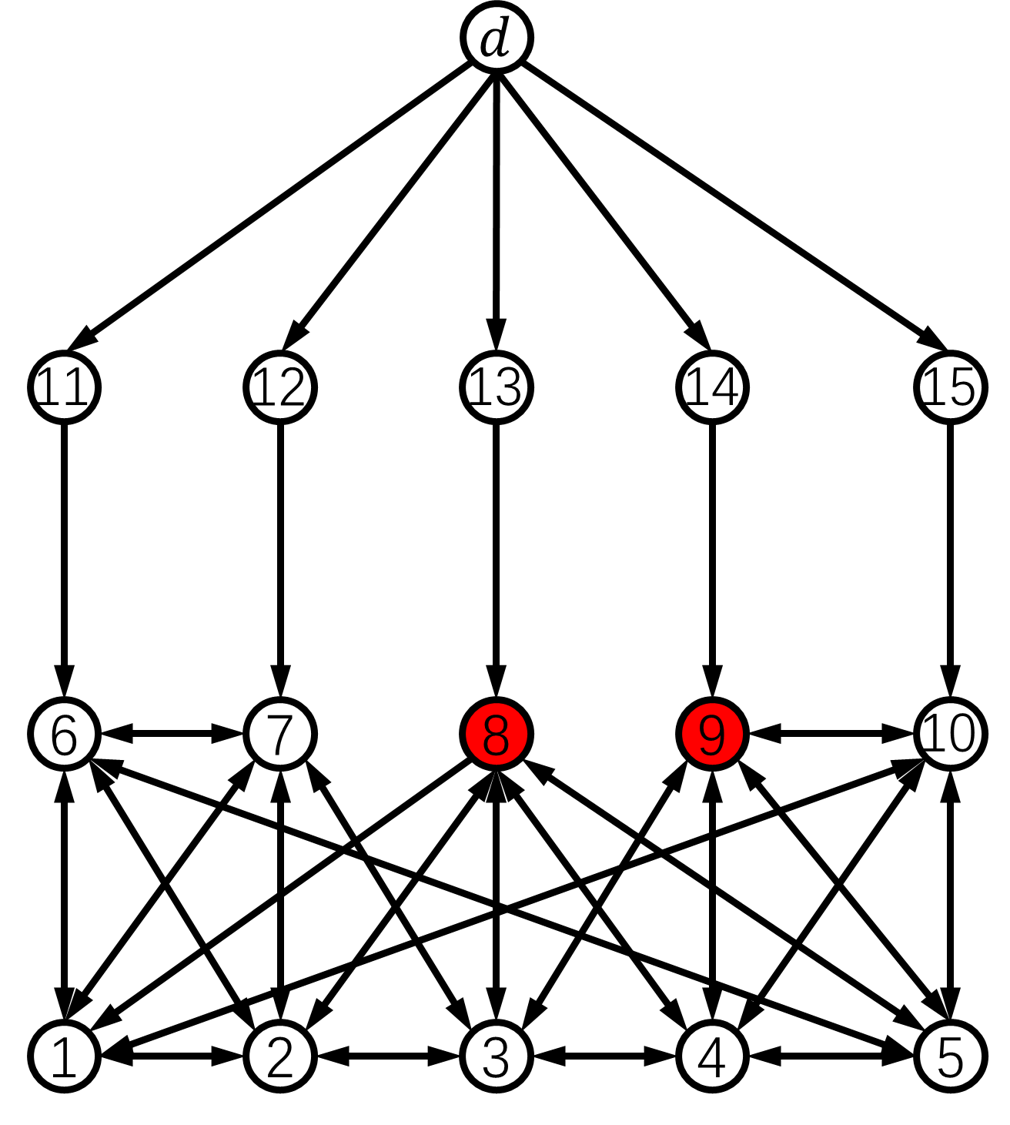}
	\vspace{-9pt}
	\caption{The graph is not a 3-robust following graph with 1 hop but is a 3-robust following graph with 3 hops under the 2-local model.}
	\label{15node}
	%\vspace*{-1.0mm}
\end{figure}

Recently, \cite{zegers2022event} studied dynamic leader-follower consensus. Their approach requires that the undirected normal network\footnote{For a network $\mathcal{G} = (\mathcal{V}, \mathcal{E})$, the normal network is the subgraph induced by normal nodes, denoted by $\mathcal{G}_{\mathcal{N}} = (\mathcal{N},\mathcal{E}_{\mathcal{N}})$.} is connected.
It is worth noting that we can obtain a tighter graph condition defined on the normal network for our leader-follower problem; the condition is that $\mathcal{G}_{\mathcal{N}}$ satisfies the property presented in Definition~\ref{robust_following} instead of $\mathcal{G}_{\mathcal{H}}$ there. In fact, if $\mathcal{G}$ is an $r$-robust following graph with $l$ hops under the $f$-local model, then $\mathcal{G}_{\mathcal{N}}$ is guaranteed to satisfy that property. However, the graph condition on the normal network cannot be verified prior to the deployment of the algorithm. Thus, we define our condition on the original network topology $\mathcal{G}$ as in the works (\cite{vaidya2012iterative,su2017reaching,yuan2021resilient}).

We further note that our results can be extended to the following two situations: (i) The edge set $\mathcal{E}$ of $\mathcal{G}$ varies in finite time; (ii) the $f$-local adversary set $\mathcal{A}$ varies in finite time (e.g., an agent may be attacked for a finite time). For example, a sufficient condition for handling the two cases is that $\mathcal{G}$ satisfies Definition~\ref{robust_following} at each time and $\mathcal{G}$ is fixed when $k\to \infty$. This condition also guarantees that $\mathcal{G}_{\mathcal{N}}$ satisfies the property in Definition~\ref{robust_following}.

\section{Resilient Dynamic Leader-Follower Consensus in First-Order MASs}\label{sec_first}

In this section, we consider first-order MASs with agents in $\mathcal{V}$ having first-order dynamics.

We denote by $e_i[k] = x_i[k] -x_d[k]$ the leader-follower consensus error of agent $i \in \mathcal{V}$ at time $k$. Then, the objective in Problem~\ref{problem} is transformed into developing a resilient control strategy such that $\forall i\in\mathcal{W}^\mathcal{N}$, $e_i[k]$ converges toward zero with ultimately bounded errors.

We introduce our first-order system model under attacks.
The leader $d$ propagates its value $x_d[k]$ to the direct followers in $\mathcal{W}_d$ at each time $k$, which is updated as
\begin{align}\label{leader1}
	 x_d[k+1] = x_d[k] + T u_d [k] ,
\end{align}
where $u_d[k]=v_d[k]$ is the bounded input with known bound (in Assumption~\ref{leader-speed}) at time $k$.

Each follower agent $i\in\mathcal{W}$ applies a control rule consisting of its neighbors’ values to make updates by
\begin{align}\label{system1}
	 x_i[k+1] = x_i[k] + T u_i [k] + T \sigma_i [k],
\end{align} 
where $x_i [k]$ and $\sigma_i [k]$ represent, respectively, the state and the bounded external disturbance with known bound $\overline{\sigma}_i$ as $|\sigma_i [k]| \leq \overline{\sigma}_i$ at each time $k$. 
Moreover, the control input $u_i [k]$ is given as
\begin{align}\label{input1}
	u_i [k] = -\gamma_i[k] \phi_i [k] - \theta_i \textup{sat}_\epsilon(\phi_i [k] ), \medspace\medspace \forall i \in \mathcal{W}^\mathcal{N}.
\end{align}
Here, $\gamma_i[k]= \alpha_i / \sum_{j =1}^{n} a_{ij}[k]$, and $\theta_i, \epsilon  \in \mathbb{R}_{>0}$. Moreover, $\phi_i [k]$ is from \eqref{msrupdate} in Algorithm~1 and is rewritten as
\begin{align}\label{phiik2}
	\phi_i [k] = \sum_{j =1}^{n} a_{ij}[k] (x_i[k] - x_j[k]),
\end{align}
where $a_{ij}[k] > 0$ if $m_{ji}[k]\in \mathcal{M}_i[k]\setminus \mathcal{R}_i[k]$, and $a_{ij}[k] = 0$  otherwise.
Note that the parameter $\theta_i$ in \eqref{input1} is designed to compensate the velocity of the leader. If the control input does not contain such a term (e.g., \cite{usevitch2020resilient,yuan2024reaching}), it could be difficult for followers to track the leader's dynamic value, i.e., the consensus error would be large. See \cite{ren2007multi} for related discussions on the performance of fault-free consensus protocols without the term for compensating leader's velocity.
However, the control inputs of adversary nodes in $\mathcal{A}$ are arbitrary and may deviate from \eqref{input1}.

For the parameters in \eqref{input1} and $T$, we assume 
\begin{align}
	& \overline{v}_d + \overline{\sigma}_i  < \theta_i , \label{theta} \\ 
	& \overline{v}_d + \overline{\sigma}_i  + \theta_i \leq \frac{2-2\alpha_i T}{T} \epsilon .  \label{theta_upper}
\end{align}

\begin{update}\label{updaterule1}
	At each time $k\geq 0$, each normal direct follower $j\in \mathcal{W}_d \cap \mathcal{N}$ updates its value as $\phi_j [k] = x_j[k] - x_d[k]$ and then follows \eqref{system1} with \eqref{input1}.
	Moreover, each normal non-direct follower $i\in \mathcal{W}^\mathcal{N} \setminus \mathcal{W}_d$ updates its $\phi_i [k]$ according to Algorithm~1, and follows \eqref{system1} with \eqref{input1}.
\end{update}

\subsection{Convergence Analysis}

The following theorem states a necessary and sufficient condition for directed networks using Update Rule~\ref{updaterule1} to achieve resilient dynamic leader-follower consensus.

Let $N =|\mathcal{W}^\mathcal{N}|$. The consensus error bound is given by
\begin{align}\label{errorbound}
	\overline{\epsilon} = \sum_{h =1}^{N} \epsilon_h, 
\end{align}
where $\epsilon_1=\epsilon$ and $\epsilon_h=\epsilon + \sum_{q =1}^{h-1} q \epsilon_q, h=2,\dots, N$. Equivalently, $\epsilon_h=h \epsilon_{h-1}, h=2,\dots, N$.

%We now state the main result for Update Rule~\ref{updaterule1}.

\begin{theorem}\label{theorem_firstorder}
	Consider the network $\mathcal{G} = (\mathcal{V},\mathcal{E})$ with $l$-hop communication, where each normal follower node $i\in \mathcal{W}^\mathcal{N}$ updates its value according to Update Rule~\ref{updaterule1}. Under Assumptions~\ref{secured}--\ref{assumption_path} and the $f$-local Byzantine set $\mathcal{A}$, resilient dynamic leader-follower consensus is achieved as in \eqref{reach_consensus} with consensus error bound $\overline{\epsilon}$ in \eqref{errorbound} if and only if $\mathcal{G}$ is an $(f+1)$-robust following graph with $l$ hops. 
\end{theorem}

\begin{proof}
	\textit{(Necessity)} If $\mathcal{G}$ is not an $(f+1)$-robust following graph with $l$ hops, then by Definition~\ref{robust_following}, 
	there exists an $f$-local set $\mathcal{F}$ such that $\mathcal{G}_{\mathcal{H}}$ does not satisfy the condition there. Suppose that $\mathcal{F}$ is exactly the set of Byzantine agents $\mathcal{A}$. Then in the normal network $\mathcal{G}_{\mathcal{N}} = (\mathcal{N},\mathcal{E}_{\mathcal{N}})$, there must be a nonempty subset $\mathcal{S}\subseteq \mathcal{W}^\mathcal{N} \setminus \mathcal{W}_d$ such that $\mathcal{Z}_{\mathcal{S}}^{f+1} = \emptyset$. It further means that
	\begin{align}\label{atmostf}
		|\mathcal{I}_{i, \mathcal{S}}^l| \leq f , \medspace \forall i\in \mathcal{S}.
	\end{align} 
	Suppose that $x_i[0]=a , \medspace \forall i\in \mathcal{S}$, and $x_j[0]=x_d[0], \medspace \forall j\in \mathcal{N}\setminus \mathcal{S}$, where $a< x_d[0]$ is a constant and the leader $d$ increases its value $x_d[k]$ as time evolves. Moreover, suppose that all Byzantine nodes send $a$ and $x_d[k]$ to the nodes in $\mathcal{S}$ and $ \mathcal{N}\setminus \mathcal{S}$, respectively, at each time $k$. 
	For normal node $i \in \mathcal{S}$, \eqref{atmostf} indicates that the cardinality of the MMC of the values larger than its own value (i.e., values from the normal nodes outside of $\mathcal{S}$) is at most $f$. These values are disregarded by Algorithm~1. Moreover, since the Byzantine nodes send $a$ to node $i$, it will use these values. Thus, node $i$ will keep its value $a$ at all times and the leader will have an increasing value, i.e., the dynamic leader-follower consensus cannot be achieved.

	\textit{(Sufficiency)}
	Recall that $\mathcal{N}=\{d\} \cup \mathcal{W}^\mathcal{N}$ and $|\mathcal{W}^\mathcal{N}|=N$.
	Sort $e_i[k]$, $i\in \mathcal{N}$, from the smallest to the largest values and rename them as $\delta_1[k], \delta_2[k], \dots, \delta_{N+1}[k]$ with
	\begin{align*}
		\delta_1[k]\leq \delta_2[k]\leq \cdots \leq \delta_{N+1}[k] .
	\end{align*}
	Define the following two sets for $h=1, 2,\dots, N+1$:
	\begin{align*}
		\overline{\mathcal{Z}}_h(k)&=\{\delta_{N+2-h}[k], \dots, \delta_{N}[k], \delta_{N+1}[k] \},\\[1mm]
		\underline{\mathcal{Z}}_h(k)&=\{\delta_1[k], \delta_2[k], \dots , \delta_{h}[k] \}.
	\end{align*}
	At any time $k \geq 0$ and for $h=1, 2,\dots, N+1$, the following node sets are defined as 
	\begin{align}\label{xhk}
		\overline{\mathcal{X}}_h(k)&=\{i\in \mathcal{N}  :  e_i[k] \in \overline{\mathcal{Z}}_h(k) \}, \nonumber \\[1mm]
		\underline{\mathcal{X}}_h(k)&=\{i\in \mathcal{N}  :  e_i[k] \in \underline{\mathcal{Z}}_h(k) \} . 
	\end{align}
	Since $e_i[k] = x_i[k] -x_d[k]$, $i \in \mathcal{W}$, and $e_d[k]  = 0$, one gets $\delta_1[k] \leq 0$ and $\delta_{N+1}[k] \geq 0$. Therefore, if $\delta_1[k] = \delta_{N+1}[k]$, one obtains $\delta_1[k] = \delta_{N+1}[k]=0$, indicating that all the normal followers successfully track the leader. 
	
	In what follows, we analyze the behaviors of the normal followers in three steps when $\delta_1[k] \neq \delta_{N+1}[k]$. First, we present equations describing the evolution of consensus errors. Second, we show that $\delta_{N+1}[k]-\delta_{N}[k]$ and $\delta_{2}[k]-\delta_{1}[k]$ are uniformly ultimately bounded.  
	Third, we show by induction that $\delta_{N+1}[k]$ and $\delta_1[k]$ are uniformly ultimately bounded, which leads us to \eqref{reach_consensus} and \eqref{errorbound}. 
	
	%\vspace{-4mm} 
	\textbf{Step 1:}
	From \eqref{system1} and \eqref{input1}, $\forall i \in \mathcal{W}^\mathcal{N}$, we have
	\begin{align}
		\frac{x_i[k+1]-x_i[k]}{T}= -\gamma_i[k] \phi_i [k] - \theta_i \textup{sat}_\epsilon(\phi_i [k] ) + \sigma_i [k].
	\end{align}
	It further follows that
	\begin{align}\label{eik}
		\frac{e_i[k+1]-e_i[k]}{T} &= -\gamma_i[k] \phi_i [k] - \theta_i \textup{sat}_\epsilon(\phi_i [k] ) + \sigma_i [k] \nonumber \\
		&\hspace*{0.5cm}\mbox{}  - \frac{x_d[k+1]-x_d[k]}{T},
	\end{align}
	where $\phi_i [k]$ is from \eqref{phiik2} and can be rewritten as
	\begin{align}\label{phiik}
		\phi_i [k] = \sum_{j =1}^{n} a_{ij}[k] (e_i[k] - e_j[k]).
	\end{align}
	For the saturation function in \eqref{sat}, when $|\phi_i [k]| > \epsilon$, it holds that $\textup{sat}_\epsilon(\phi_i [k] ) = \textup{sgn}(\phi_i [k])$. Moreover, since $\sigma_i [k]$ and $|x_d[k+1]-x_d[k]|/T$ are bounded, there exists a bounded number $\eta_i [k]$ such that \eqref{eik} can be rewritten as
	\begin{align}\label{eik1}
		\frac{e_i[k+1]-e_i[k]}{T} = -\gamma_i[k] \phi_i [k] - \eta_i[k] \textup{sat}_\epsilon(\phi_i [k] ) .
	\end{align}
	Since $\textup{sat}_\epsilon(\phi_i [k] ) =\textup{sgn}(\phi_i [k])$, when both $\sigma_i [k]$ and $- (x_d[k+1]-x_d[k])$ in \eqref{eik} have the inverse sign as $\phi_i [k]$, we have $\eta_i[k] =\theta_i + |x_d[k+1]-x_d[k]|/T + |\sigma_i [k]|$,
	and when both $\sigma_i [k]$ and $- (x_d[k+1]-x_d[k])$ have the same sign as $\phi_i [k]$, one gets $\eta_i[k] =\theta_i - |x_d[k+1]-x_d[k]|/T - |\sigma_i [k]|>0$ by \eqref{theta}. Thus, at each time $k$, the two values are the possible maximum and minimum values of $\eta_i[k]$. Thus, $\eta_i[k]$ satisfies
	\begin{align}\label{eta1}
		0 &< \theta_i - \sup_{k'\geq0}\left\{ \frac{|x_d[k'+1]-x_d[k']|}{T}+|\sigma_i [k']| \right\} \leq \eta_i[k] \nonumber \\
		&\leq \theta_i + \sup_{k'\geq0}\left\{ \frac{|x_d[k'+1]-x_d[k']|}{T}+|\sigma_i [k']| \right\}.
	\end{align} 
	
	\vspace{-4mm}
	\textbf{Step 2:}
	Since $\delta_1[k]\neq \delta_{N+1}[k]$, at least one normal follower $i\in \mathcal{W}^\mathcal{N}$ with $e_i[k]\neq 0$ satisfies
	\begin{align}\label{x1}
		i\in &\overline{\mathcal{X}}_1(k)  \medspace\medspace \textup{and} \medspace\medspace d\notin \overline{\mathcal{X}}_1(k), \nonumber \\[1mm]
	\textup{or} \medspace\medspace	i\in &\underline{\mathcal{X}}_1(k) \medspace\medspace \textup{and} \medspace\medspace d\notin \underline{\mathcal{X}}_1(k).
	\end{align}
	There are two cases for such normal followers. The first case is $\exists j\in \mathcal{W}_d$ satisfying \eqref{x1}, and the second case is $\nexists j\in \mathcal{W}_d$ satisfying \eqref{x1}. They are analyzed separately:
	
	(i) $\exists j\in \mathcal{W}_d$ satisfying \eqref{x1}: According to the definition of $\mathcal{W}_d$ and Update Rule~\ref{updaterule1}, there must exist a normal follower $i\in \mathcal{W}^\mathcal{N}$ either in $\overline{\mathcal{X}}_1(k)$ or $\underline{\mathcal{X}}_1(k)$ which makes an update using a value from the normal nodes outside the set to which it belongs.
	
	(ii) $\nexists j\in \mathcal{W}_d$ satisfying \eqref{x1}: In this case, since $\mathcal{G}$ is an $(f+1)$-robust following graph with $l$ hops and the adversarial set $\mathcal{A}$ is an $f$-local set, the normal network $\mathcal{G}_{\mathcal{N}} = (\mathcal{N},\mathcal{E}_{\mathcal{N}} )$ must satisfy that for every nonempty subset $\mathcal{S}\subseteq \mathcal{W}^\mathcal{N} \setminus \mathcal{W}_d $, the following condition holds:
	\begin{align*}
		| \mathcal{Z}_{\mathcal{S}}^{f+1}| 
		=| \{i\in \mathcal{S} :  |\mathcal{I}_{i, \mathcal{S}}^l| \geq f+1 \} | \geq 1  .
	\end{align*} 
	Hence, for nonempty node set $\overline{\mathcal{X}}_1(k)$ or $\underline{\mathcal{X}}_1(k)$ in $\mathcal{G}_{\mathcal{N}}$, there exists a normal follower $i$ such that 
	\begin{align*}
		& |\mathcal{I}_{i, \overline{\mathcal{X}}_1(k) }^l | \geq f+1, \medspace\medspace\medspace \textup{if} \medspace i\in \overline{\mathcal{X}}_1(k),  \\
		& |\mathcal{I}_{i, \underline{\mathcal{X}}_1(k) }^l | \geq f+1, \medspace\medspace\medspace \textup{if} \medspace i\in \underline{\mathcal{X}}_1(k)  .
	\end{align*}
	Therefore, there must exist a normal follower $i\in \mathcal{W}^\mathcal{N}$ either in $\overline{\mathcal{X}}_1(k)$ or $\underline{\mathcal{X}}_1(k)$ which makes an update using a value from the normal nodes outside the set to which it belongs. This can be seen from step 2 of the MW-MSR algorithm where node $i$ can only remove smaller or larger values from at most $f$ nodes.

	According to \eqref{xhk}, if $i\in\overline{\mathcal{X}}_1(k)$ or $i\in\underline{\mathcal{X}}_1(k)$, then it holds that $e_i[k] = \delta_{N+1}[k]$ or $e_i[k] = \delta_{1}[k]$, respectively.  
	In this condition, the following situations can happen, respectively, when $|\overline{\mathcal{X}}_1(k)| = 1$ if $i\in\overline{\mathcal{X}}_1(k)$ and when $|\underline{\mathcal{X}}_1(k)| = 1$ if $i\in\underline{\mathcal{X}}_1(k)$: 
	\begin{align}\label{deltan+1}
		\delta_{N+1}[k]-\delta_{N}[k] &> \epsilon_1, \medspace\medspace\medspace \textup{if} \medspace i\in \overline{\mathcal{X}}_1(k), \nonumber \\
		\delta_{2}[k]-\delta_{1}[k] &> \epsilon_1, \medspace\medspace\medspace \textup{if} \medspace i\in \underline{\mathcal{X}}_1(k)  .
	\end{align}
	Moreover, we have known that there is a normal follower $i$ either in $\overline{\mathcal{X}}_1(k)$ or $\underline{\mathcal{X}}_1(k)$ which makes an update using a value from the normal nodes outside the set to which it belongs.  
	Thus, under \eqref{deltan+1}, we can conclude that the distance of $e_i[k]$ from the closest normal $e_j[k]$, is larger than $\epsilon_1$. Moreover, since $\mathcal{A}$ is $f$-local, any adversarial neighbor with consensus error outside the interval $[\delta_{1}[k], \delta_{N+1}[k]]$ will be disregarded by the MW-MSR algorithm. Therefore, according to \eqref{phiik} and \eqref{deltan+1}, for any possible $a_{ij}[k]\geq0$ used in \eqref{phiik}, we have
	\begin{align}\label{phiikabsolute}
		\phi_i[k]> \epsilon, \medspace\medspace\medspace &\textup{if} \medspace i\in \overline{\mathcal{X}}_1(k),  \nonumber \\
		\phi_i[k]< -\epsilon, \medspace\medspace\medspace &\textup{if} \medspace i\in \underline{\mathcal{X}}_1(k)  .
	\end{align}
	Since $|\phi_i[k]| > \epsilon$, we further have $\textup{sat}_\epsilon(\phi_i [k] ) = \textup{sgn}(\phi_i[k])$. Then, from \eqref{x1} and \eqref{phiikabsolute}, it holds that
	\begin{align*}
		\textup{sat}_\epsilon(\phi_i [k] ) = \textup{sgn}(\phi_i[k])=\textup{sgn}(e_i[k]) .
	\end{align*}

	Now, for $|\phi_i[k]| > \epsilon$, we consider the following Lyapunov candidate:
	\begin{align}\label{v}
		V_i[k] = e_i^2[k] .
	\end{align}
	It further implies 
	\begin{align*}
		\Delta V_i[k] &= V_i[k+1]-V_i[k] \\
		%&= e_i^2[k+1]- e_i^2[k]\\
		&= (e_i[k+1]+ e_i[k])(e_i[k+1]- e_i[k])\\
		&= \left(2e_i[k]-\gamma_i[k] \phi_i [k]T - \eta_i[k] \textup{sat}_\epsilon(\phi_i [k] )T\right)\\
		&\hspace*{0.5cm}\mbox{}  \times (e_i[k+1]- e_i[k]).
	\end{align*}
	Notice from \eqref{eik1} that sgn$(e_i[k])=-\textup{sgn}(e_i[k+1]-e_i[k])$.
	Therefore, in the following, we will show that
	\begin{align}\label{sgn}
		\textup{sgn}\left(2e_i[k]-\gamma_i[k] \phi_i [k]T - \eta_i[k] \textup{sat}_\epsilon(\phi_i [k] )T\right) =\textup{sgn}(e_i[k]).
	\end{align}
	To this end, it is equivalent to show that
	\begin{align*}
		|2e_i[k]| \geq |\gamma_i[k] \phi_i [k]|T + |\eta_i[k] |T.
	\end{align*}
	Suppose $|\delta_{N+1}[k]| \geq |\delta_{1}[k]|$, then consider $i\in\overline{\mathcal{X}}_1(k)$. 
	We can observe that $|\gamma_i[k] \phi_i [k]|\leq 2\alpha_i|e_i[k]|$. Thus, we need
	\begin{align*}
		0< |\eta_i[k] | < \frac{2-2\alpha_i T}{T} |e_i[k]|.
	\end{align*}
	Moreover, since $i\in\overline{\mathcal{X}}_1(k)$, it must hold that $|e_i[k]| > \epsilon$.
	Thus, by \eqref{theta_upper} and \eqref{eta1}, we have
	\begin{align}\label{eta2}
		0< |\eta_i[k] | < \frac{2-2\alpha_i T}{T} \epsilon,
	\end{align}
	and then we can obtain
	\begin{align}\label{dv}
		\Delta V_i[k] = V_i[k+1]-V_i[k] <0.
	\end{align}
	Similarly, if $|\delta_{N+1}[k]| \leq |\delta_{1}[k]|$, then consider $i\in\underline{\mathcal{X}}_1(k)$, and \eqref{dv} follows.

	Therefore, while $|\phi_i [k]| > \epsilon$, $V_i[k]$ is decreasing. Thus, in a finite time, we should have
	\begin{align}\label{deltan+1-reverse}
		\delta_{N+1}[k]-\delta_{N}[k] &\leq \epsilon_1, \medspace\medspace\medspace \textup{if} \medspace i\in \overline{\mathcal{X}}_1(k),  \nonumber \\
		\delta_{2}[k]-\delta_{1}[k] &\leq \epsilon_1, \medspace\medspace\medspace \textup{if} \medspace i\in \underline{\mathcal{X}}_1(k)  .
	\end{align}  
	Also, if we do not encounter \eqref{deltan+1}, we have \eqref{deltan+1-reverse} directly. 
	
	%\vspace{-4mm}
	\textbf{Step 3:} 
	When $\delta_{N+1}[k]\neq \delta_{1}[k]$, at least one of the normal followers
	with nonzero consensus error belongs to $\overline{\mathcal{X}}_1(k)$ or $\underline{\mathcal{X}}_1(k)$ and
	\begin{align}
		\overline{\mathcal{X}}_1(k)&\subseteq \overline{\mathcal{X}}_h(k), \nonumber \\[1mm]
		\underline{\mathcal{X}}_1(k)&\subseteq \underline{\mathcal{X}}_h(k),\medspace h=2,\dots, N+1.
	\end{align}
	Then, in $\overline{\mathcal{X}}_h(k)$ or $\underline{\mathcal{X}}_h(k)$, there exists at least one normal follower $i\in \mathcal{W}^\mathcal{N}$ with $e_i[k]\neq 0$. Thus, if $d \notin \overline{\mathcal{X}}_h(k)\cap\underline{\mathcal{X}}_h(k)$, we consider normal followers $i$ where
	\begin{align}\label{xhk1}
		i\in &\overline{\mathcal{X}}_h(k) \medspace\medspace \textup{and} \medspace\medspace d\notin \overline{\mathcal{X}}_h(k), \nonumber \\[1mm]
		\textup{or} \medspace\medspace  i\in &\underline{\mathcal{X}}_h(k) \medspace\medspace \textup{and} \medspace\medspace d\notin \underline{\mathcal{X}}_h(k).
	\end{align} 
	Since $\mathcal{G}$ is an $(f+1)$-robust following graph with $l$ hops, there is at least one node in $\overline{\mathcal{X}}_h(k)$ or $\underline{\mathcal{X}}_h(k)$ which makes an update using at least one value from the normal nodes outside the set to which it belongs.
	Thus, the following situations can happen, respectively, when $|\overline{\mathcal{X}}_h(k)| = h$ if $i\in\overline{\mathcal{X}}_h(k)$ and when $|\underline{\mathcal{X}}_h(k)| = h$ if $i \in \underline{\mathcal{X}}_h(k)$:
	\begin{align}\label{epsilon_h}
		\delta_{N+2-h}[k]-\delta_{N+1-h}[k] &> \epsilon_h, \medspace\medspace\medspace \textup{if} \medspace i\in \overline{\mathcal{X}}_h(k),  \nonumber \\ 
		\delta_{h+1}[k]-\delta_{h}[k] &> \epsilon_h, \medspace\medspace\medspace \textup{if} \medspace i\in \underline{\mathcal{X}}_h(k).
	\end{align}
	Here, we discuss why \eqref{epsilon_h} can happen.
	Because of the ultimate error bounds obtained in steps 1 to $h-1$ (i.e., $\epsilon_1, \epsilon_2,\dots, \epsilon_{h-1}$) and by the definition of $\epsilon_h$, for a finite step $h$, we can obtain
	\begin{align*}
		\delta_{N+2-h}[k]-\delta_{N+1-h}[k] > \medspace\medspace\medspace\medspace\medspace\medspace\medspace\medspace\medspace\medspace\medspace\medspace\medspace\medspace\medspace\medspace\medspace\medspace\medspace\medspace\medspace\medspace\medspace\medspace\medspace\medspace\medspace  &\\
		\sum_{j=1}^{h-1}\left(\delta_{N+2-j}[k]-\delta_{N+1-h}[k]\right)+ \epsilon&, \medspace \textup{if} \medspace i\in \overline{\mathcal{X}}_h(k),  \\
		\delta_{h+1}[k]-\delta_{h}[k] > \sum_{j=1}^{h-1}\left(\delta_{h}[k]-\delta_{j}[k]\right)+ \epsilon&, \medspace \textup{if} \medspace i\in \underline{\mathcal{X}}_h(k).
	\end{align*} 
	This implies that if $i\in \underline{\mathcal{X}}_h(k)$ and $d\notin \underline{\mathcal{X}}_h(k)$, the distance of the consensus error $\delta_{h}[k]$ from the error $\delta_{h+1}[k]$ is larger than the summation of its distances from all the possible smaller errors inside $\underline{\mathcal{X}}_h(k)$ plus $\epsilon$. The discussion is similar when $i\in \overline{\mathcal{X}}_h(k)$ and $d\notin \overline{\mathcal{X}}_h(k)$. Thus, we have proved that the situations in \eqref{epsilon_h} can hold for the two cases in \eqref{xhk1}, respectively. 
	
	Moreover, normal followers using the MW-MSR algorithm will discard the values with consensus errors outside the interval $[\delta_{1}[k], \delta_{N+1}[k]]$. Hence, from \eqref{phiik} and the fact that at least one normal follower in $\overline{\mathcal{X}}_h(k)$ or $\underline{\mathcal{X}}_h(k)$ updates its value using at least one value from the normal node outside the set to which it belongs, we have
	\begin{align}\label{phiikabsoluteh}
		\phi_i[k]> \epsilon, \medspace\medspace\medspace &\textup{if} \medspace i\in \overline{\mathcal{X}}_h(k),  \nonumber \\
		\phi_i[k]< -\epsilon, \medspace\medspace\medspace &\textup{if} \medspace i\in \underline{\mathcal{X}}_h(k)  .
	\end{align}
	Since $|\phi_i[k]| > \epsilon$, we have $\textup{sat}_\epsilon(\phi_i [k] ) = \textup{sgn}(\phi_i[k])$. Then, from \eqref{xhk1} and \eqref{phiikabsoluteh}, it holds that
	\begin{align*}
		\textup{sat}_\epsilon(\phi_i [k] ) = \textup{sgn}(\phi_i[k])=\textup{sgn}(e_i[k]) .
	\end{align*}
	Furthermore, using the Lyapunov candidate $V_i[k]$ in \eqref{v}, we have $\Delta V_i[k]$ as in \eqref{dv} by the reasoning similar to that in Step 2. As a result, while $|\phi_i[k]| > \epsilon$, $V_i[k]$ is decreasing and the convergence is in a finite time. By considering the ultimate error bounds obtained
	in steps $1$ to $h-1$, for step $h\geq 1$, the following should hold in a finite time:
	\begin{align}\label{stepbound}
		\delta_{N+2-h}[k]-\delta_{N+1-h}[k] &\leq \epsilon_h, \medspace\medspace\medspace \textup{if} \medspace i\in \overline{\mathcal{X}}_h(k), \nonumber  \\
		\delta_{h+1}[k]-\delta_{h}[k] &\leq \epsilon_h, \medspace\medspace\medspace \textup{if} \medspace i\in \underline{\mathcal{X}}_h(k)  .
	\end{align}
	Finally, if there exists a step $1 < h \leq N + 1$ such that $d \in \overline{\mathcal{X}}_h(k)\cap\underline{\mathcal{X}}_h(k)$; then, from \eqref{stepbound}, it holds that
	\begin{align*}
		|x_i[k]-x_d[k] | &\leq \sum_{j =1}^{h-1} \epsilon_j, \medspace \forall i\in \mathcal{W}^\mathcal{N}. 
	\end{align*}
	As the maximum of $h$ such that $d \in \overline{\mathcal{X}}_h(k)\cap\underline{\mathcal{X}}_h(k)$ is $N + 1$,
	we obtain the error bound in \eqref{errorbound}. The proof is complete.
\end{proof}

For the spacial case where all followers are direct followers, i.e., $\mathcal{W}_d=\mathcal{W}$, we could obtain a smaller consensus error bound, as stated in the following lemma.

\begin{lemma}\label{lemma_firstorder}
	Consider the network $\mathcal{G} = (\mathcal{V},\mathcal{E})$ with the problem settings and assumptions in Theorem~\ref{theorem_firstorder}. Resilient dynamic leader-follower consensus is achieved as in \eqref{reach_consensus} with consensus error bound $\overline{\epsilon}_{d1}=\max \{\frac{T}{2-\alpha_i T} \eta_i[k], \medspace \epsilon \}$ if all followers are in $\mathcal{W}_d$, i.e., $\mathcal{W}_d=\mathcal{W}$.
\end{lemma}

\begin{proof}
	For any node $i\in \mathcal{W}^\mathcal{N} \subseteq \mathcal{W}_d$, if $e_i[k]>\epsilon$, then $\phi_i[k]=e_i[k]>\epsilon$ and $\gamma_i[k]=\alpha_i$. We further have 	
	\begin{align*}
		\Delta V_i[k] 
		&= \left(2e_i[k]-\gamma_i[k] \phi_i [k]T - \eta_i[k] \textup{sat}_\epsilon(\phi_i [k] )T\right)\\
		&\hspace*{0.5cm}\mbox{}  \times (e_i[k+1]- e_i[k])\\
		&= \left(2e_i[k]-e_i [k] \alpha_i T - \eta_i[k] \textup{sgn}(e_i [k] )T\right)\\
		&\hspace*{0.5cm}\mbox{}  \times (e_i[k+1]- e_i[k])
	\end{align*}
	Thus, if it holds that $\alpha_i T<2$ and
	\begin{align*}
		\textup{sgn}\left(2e_i[k]-e_i [k]\alpha_i T - \eta_i[k] \textup{sgn}(e_i [k] )T\right) =\textup{sgn}(e_i[k]),
	\end{align*}
	then by following a similar analysis, we have $\Delta V_i[k]  <0$. Therefore, when
	\begin{align*}
		|e_i[k] | \geq \frac{T}{2-\alpha_i T} \eta_i[k],
	\end{align*}
	we have $\Delta V_i[k]  <0$, and $e_i[k]$ will decrease until $e_i[k]\leq \epsilon$ if $\epsilon \geq \frac{T}{2-\alpha_iT} \eta_i[k]$. The situation is similar when $e_i[k]<-\epsilon$.
	Therefore, we conclude that in a finite time, we have 
	\begin{align*}
		|e_i[k] | \leq \max \left\{\frac{T}{2-\alpha_i T} \eta_i[k], \medspace \epsilon \right\}=\overline{\epsilon}_{d1}, \forall i\in \mathcal{W}^\mathcal{N} \subseteq \mathcal{W}_d. 
	\end{align*}
	\vspace{-2mm}
\end{proof}		

\begin{remark}\label{smallerbound1}
	Note that when $T$ is small, we could get a smaller bound for consensus errors of normal followers in $\mathcal{W}_d$ if $\frac{T}{2-\alpha_i T} \eta_i[k] < \epsilon$. Furthermore, when $|\mathcal{W}_d\cap \mathcal{N}|<N$, we could have a smaller consensus error bound $\overline{\epsilon}_1<\overline{\epsilon}$ in \eqref{errorbound} for Theorem~\ref{theorem_firstorder}. It is given by
	\begin{align}\label{errorbound_smaller}
		\overline{\epsilon}_1 = \sum_{h =1}^{N-|\mathcal{W}_d\cap \mathcal{N}|} \epsilon_h + \max \left\{\frac{T}{2-\alpha_i T} \eta_i[k], \medspace \epsilon \right\}   .
	\end{align}
\end{remark}

\vspace{-5mm} 
One can observe that the error bound in \eqref{errorbound} is conservative as it is proportional to the number of normal followers $N$; a similar bound is also reported in \cite{rezaee2021resiliency}. Our analysis can obtain a smaller error bound in \eqref{errorbound_smaller} as it is proportional to the number of normal non-direct followers $N-|\mathcal{W}_d\cap \mathcal{N}|$. Part of the reason for the conservatism is that the error will increase by $\epsilon_h$ using the MSR algorithms for each layer of normal followers (see Step 3 in the proof of Theorem~\ref{theorem_firstorder}); it further accumulates over multiple layers of normal followers. Moreover, because of the above reason, we can obtain consensus errors that are further smaller than the bound in \eqref{errorbound_smaller} for general networks, as the number of layers w.r.t. the leader will decrease with multi-hop relays.
Lastly, we note that even for the fault-free discrete-time dynamic leader-follower consensus (\cite{cao2009distributed}), the consensus error cannot approach zero unless $T\to 0$. In our results, the consensus error can approach zero if $T$ and $\epsilon$ are properly chosen such that \eqref{theta_upper} is satisfied.
	
%; this requires $T\to 0$ and $\epsilon\to 0$.}

	It is notable that our graph condition in Theorem~\ref{theorem_firstorder} is both necessary and sufficient for our algorithm to solve Problem~\ref{problem}. In fact, even for the one-hop case, it is tighter than the one in \cite{rezaee2021resiliency} as shown next in Lemma~\ref{tighter1}. In particular, \cite{rezaee2021resiliency} have derived a sufficient graph condition for a modified W-MSR algorithm, that is, $\mathcal{G}$ is a $(2f +1)$-robust leader-follower graph. It requires $|\mathcal{W}_d|\geq 2f +1$ and any nonempty set $\mathcal{S} \subseteq \mathcal{W}\setminus \mathcal{W}_d$ is $(2f +1)$-reachable, i.e., $\exists i \in \mathcal{S}$ s.t. $|\mathcal{N}_i^{1-} \setminus \mathcal{S}| \geq 2f+1$. The difference between the two conditions is shown in the next lemma, whose proof can be found in \cite{yuan2024reaching}.

\begin{lemma}\label{tighter1}
	If $ \mathcal{G}$ is a $(2f +1)$-robust leader-follower graph, then $\mathcal{G}$ is an $(f+1)$-robust following graph with $1$ hop, and the converse does not hold.
\end{lemma}

\begin{remark}
	We compare our graph condition with the ones in related works.
	With one-hop communication (i.e., no relays), our condition in Theorem~\ref{theorem_firstorder} is equivalent to the necessary and sufficient condition for the CPA to succeed (\cite{tseng2015broadcast}). Moreover, our condition with $1$ hop is tighter than the sufficient conditions in \cite{usevitch2020resilient,rezaee2021resiliency} studying resilient leader-follower consensus as discussed above. Besides, for the multi-hop case, our condition with $l\geq 2$ hops is even tighter than the conditions in related works (\cite{usevitch2020resilient,rezaee2021resiliency,tseng2015broadcast}). The reason is that the graph robustness generally increases (and definitely does not decrease) as the relay range $l$ increases. Hence, without changing the network topology, our methods can tolerate more adversaries; see, e.g., the simulations in Section~\ref{sec_example}.
\end{remark}

Next, we explain the reasons why our approach has a tighter graph condition. First, for the one-hop case, by Theorem~\ref{theorem_firstorder}, our condition is necessary for MSR-based algorithms to achieve leader-follower consensus. Here, we show the gap between the condition in \cite{rezaee2021resiliency} and ours using the graph in Fig.~\ref{9node}(b). Let $\mathcal{A}=\{5\}$ and $\mathcal{S}=\{1,2,3\}$, which does not meet the condition in \cite{rezaee2021resiliency}.
However, node 2 in $\mathcal{S}$ can still move towards the reference when its error becomes the smallest (or largest) one in $\mathcal{N}$. Then, it removes only the largest (or smallest) errors from $f$ nodes and hence can be affected by the normal nodes outside $\mathcal{S}$.
Second, for the multi-hop case, our condition is even tighter since normal nodes can obtain more values from normal neighbors. Algorithm~1 ensures that even with multi-hop relays, normal nodes can remove the extreme values from adversaries (i.e., the errors outside $[\delta_{1}[k], \delta_{N+1}[k]]$). Thus, it remains to show that normal node(s) in $\mathcal{S}$ can obtain more values from nodes in $\mathcal{N}\setminus\mathcal{S}$. For example, in Fig~\ref{9node}(a), when $l=1$, the set $\mathcal{S}=\{1,2,3,6\}$ has no node connected to its outside after applying Algorithm~1. However, node 2 in $\mathcal{S}$ satisfies this condition when $l\geq 2$.

\begin{remark}
	We study the Byzantine model, which is more adversarial than the malicious model in \cite{leblanc2013resilient,yuan2021resilient}. 
	However, we conclude that for our approaches to solve Problem~\ref{problem} under the malicious model, the graph condition in Theorem~\ref{theorem_firstorder} is necessary and sufficient; see Remark~\ref{remark_leaderless}.
\end{remark}

\subsection{Properties of (f+1)-Robust Following Graphs}

We list several properties of $(f+1)$-robust following graphs with $l$ hops from \cite{yuan2024reaching}.

\begin{lemma}\label{2f+1leaders}
	If graph $\mathcal{G}$ with $\mathcal{W}_d \neq \mathcal{W}$ is an $(f+1)$-robust following graph with $l$ hops under the $f$-local model, then the following hold:
	\begin{enumerate}
		\item $|\mathcal{W}_d| \geq 2f+1$.
		\item $\exists i \in \mathcal{W}^\mathcal{N}$ s.t. $|\mathcal{N}_i^{l-} \cap \mathcal{W}_d| \geq 2f+1$. 
		\item $|\mathcal{N}_i^{1-}| \geq 2f+1, \forall i \in \mathcal{W}\setminus \mathcal{W}_d$. Moreover, the minimum number of directed edges of $\mathcal{G}$ with minimum $|\mathcal{W}_d|$ is $(2f+1)(|\mathcal{W}\setminus \mathcal{W}_d|+1)$.
	\end{enumerate}
\end{lemma}

From Lemma~\ref{2f+1leaders}(3), we see that each follower node must have $2f+1$ incoming edges. Moreover, $\mathcal{G}$ with minimum $|\mathcal{W}_d|$ must have directed edges no less than $(2f+1)(|\mathcal{W}\setminus \mathcal{W}_d|+1)$; this requirement is consistent with the one reported in \cite{rezaee2021resiliency}. However, we emphasize that our condition is tighter as mentioned earlier and covers a wider range of graphs. Besides, we can utilize undirected edges to relax the heavy connectivity requirement. For example, in Fig.~\ref{15node}, the graph $\mathcal{G}$ has 33 directed/undirected edges while it needs 55 directed edges to satisfy the condition in \cite{rezaee2021resiliency}.

The following lemma shows the equivalence between our condition with $f=0$ and the one for the fault-free case (\cite{ren2007multi}); it can be derived from Definition~\ref{robust_following}.
\begin{lemma}
	The graph $\mathcal{G}$ is a $1$-robust following graph with $l\geq 1$ hops if and only if $\mathcal{G}$ has a spanning tree rooted at the leader.
\end{lemma}

\section{Resilient Dynamic Leader-Follower Consensus in Second-Order MASs}\label{sec_second}

In this section, we switch our attention to the MASs where agents have second-order dynamics.

Consider a second-order MAS with communication network $\mathcal{G}= (\mathcal{V},\mathcal{E})$. 
Each follower node $i\in \mathcal{W} $ has a double-integrator dynamics modified from \cite{bullo2009distributed,dibaji2017resilient}. Its discretized form is given as 
\begin{align}\label{secondorder}
	 \hat{x}_i[k+1] &= \hat{x}_i[k] + T v_i[k] , \nonumber \\ 
	 v_i[k+1] &= v_i[k] + T u_i[k] + T \sigma_i[k],
\end{align}
where $v_i[k] \in \mathbb{R}$ and $\sigma_i[k]$ are, respectively, the velocity and the bounded external disturbance with known bound $\overline{\sigma}_i$ of node $i$ at time $k$. Moreover, $\hat{x}_i[k]=x_i[k]- \rho_i$, where $x_i[k] \in \mathbb{R}$ is the absolute position value of node $i$ and $\rho_i \in \mathbb{R}$ is a constant representing the desired relative position value of node $i$ in a formation. For the sake of simplicity, we call $\hat{x}_i[k]$ to be the agents’ position values.
Similarly, the leader $d$ updates its value as
\begin{align}\label{secondorder-leader}
	 \hat{x}_d[k+1] &= \hat{x}_d[k] + T v_d[k] , \nonumber \\ 
	 v_d[k+1] &= v_d[k] + T u_d[k],
\end{align}
where $ \hat{x}_d[k] = x_d[k]- \rho_d$. Here, $x_d[k]$ and $\rho_d \in \mathbb{R}$ are the absolute position and the desired relative position in the formation of leader $d$, respectively.

For the second-order MAS in \eqref{secondorder}, our resilient dynamic leader-follower consensus problem is almost the same as Problem~\ref{problem} except that agents exchange $\hat{x}$ values with neighbors.
To solve this problem, we present Update Rule~\ref{updaterule2}, where each follower uses only the position values $\hat{x}_i[k]$ of neighbors within $l$ hops in its control input at each time $k$. Specifically, for each normal follower $i\in\mathcal{W}^\mathcal{N}$, the control input $u_i [k]$ is given as
\begin{align}\label{input2}
	u_i [k] = -\gamma_i[k] \phi_i [k] - \theta_i \textup{sat}_\epsilon(\phi_i [k] ) -\beta v_i[k].
\end{align}
Here, $\gamma_i[k]= \alpha_i / \sum_{j =1}^{n} a_{ij}[k]$, and $\theta_i, \epsilon,\beta  \in \mathbb{R}_{>0}$. Moreover, $\phi_i [k]$ is obtained using \eqref{msrupdate} in Algorithm~1 and is rewritten by
\begin{align}\label{phiik_second}
	\phi_i [k] = \sum_{j =1}^{n} a_{ij}[k] (\hat{x}_i[k] - \hat{x}_j[k]),
\end{align}
where $a_{ij}[k] > 0$ if $m_{ji}[k]\in \mathcal{M}_i[k]\setminus \mathcal{R}_i[k]$, and $a_{ij}[k] = 0$ otherwise. However, the control inputs of adversary nodes in $\mathcal{A}$ are arbitrary and may deviate from above.

\begin{update}\label{updaterule2}
	At each time $k\geq 0$, each normal direct follower $j\in \mathcal{W}_d \cap \mathcal{N}$ updates its value as $\phi_j [k] = \hat{x}_j[k] - \hat{x}_d[k]$ and then follows \eqref{secondorder} with \eqref{input2}.
	Moreover, each normal non-direct follower $i\in \mathcal{W}^\mathcal{N} \setminus \mathcal{W}_d$ updates its $\phi_i [k]$ according to Algorithm~1 with inputs $\hat{x}_i[k], k\geq 0$, and follows \eqref{secondorder} with \eqref{input2}.
\end{update}

For the parameters $\beta$ and $\theta_i$ in \eqref{input2}, we assume
\begin{align}\label{beta}
	0< \beta \overline{v}_d \leq \theta_i .
\end{align}

\vspace{-4mm}
\subsection{Convergence Analysis}

The next theorem is the main result of this section.
It provides a necessary and sufficient graph condition for Update Rule~\ref{updaterule2} to achieve the objective in \eqref{reach_consensus}.

\begin{theorem}\label{theorem_secondorder}
	Consider the network $\mathcal{G} = (\mathcal{V},\mathcal{E})$ with $l$-hop communication, where each normal follower node $i\in \mathcal{W}^\mathcal{N}$ updates its value according to Update Rule~\ref{updaterule2}. Under Assumptions~\ref{secured}--\ref{assumption_path} and the $f$-local Byzantine set $\mathcal{A}$, resilient dynamic leader-follower consensus is achieved as in \eqref{reach_consensus} with consensus error bound $\overline{\epsilon}$ in \eqref{errorbound} if and only if $\mathcal{G}$ is an $(f+1)$-robust following graph with $l$ hops. 
\end{theorem}

\begin{proof}
	\textit{(Necessity)} If $\mathcal{G}$ is not an $(f+1)$-robust following graph with $l$ hops, by following the same reasoning as in the necessity proof of Theorem~\ref{theorem_firstorder}, there is a nonempty subset $\mathcal{S}\subseteq \mathcal{W}^\mathcal{N} \setminus \mathcal{W}_d$ such that \eqref{atmostf} holds.
	
	Suppose that $\hat{x}_i[0]=a , \medspace \forall i\in \mathcal{S}$, and $\hat{x}_j[0]=\hat{x}_d[0], \medspace \forall j\in \mathcal{N}\setminus \mathcal{S}$, where $a< \hat{x}_d[0]$ is a constant and the leader $d$ increases its value $\hat{x}_d[k]$ as time evolves. Moreover, let $v_i[0]=0, \medspace \forall i\in \mathcal{V}$. Assume that at each time $k$, all Byzantine nodes send $a$ and $\hat{x}_d[k]$ to the nodes in $\mathcal{S}$ and $ \mathcal{N}\setminus \mathcal{S}$, respectively. 
	For any normal node $i \in \mathcal{S}$, by Algorithm~1, it removes all the values of neighbors outside $\mathcal{S}$ since the cardinality of the MMC of these values is equal to $f$ or less. Then, according to Update Rule~\ref{updaterule2}, such normal nodes will keep their values, i.e., $\hat{x}_i[k]=a, v_i[k]=0, \medspace \forall i\in \mathcal{S}, \medspace \forall k\geq 0$. Meanwhile, the value of the leader is increasing. Thus, resilient dynamic leader-follower consensus cannot be achieved.

	\textit{(Sufficiency)}
	The sufficiency follows a similar proof as that of Theorem~\ref{theorem_firstorder}. Hence, we provide the crucial parts.
	
	Let $\hat{e}_i[k] = \hat{x}_i[k] -\hat{x}_d[k], \forall i\in \mathcal{V}$. Sort $\hat{e}_i[k]$, $i\in \mathcal{N}$, from the smallest to the largest values and denote them as $\delta_1[k], \delta_2[k], \dots, \delta_{N+1}[k]$. Define the node sets $\overline{\mathcal{X}}_h(k)$ and $\underline{\mathcal{X}}_h(k)$ as in \eqref{xhk} using $\hat{e}_i[k]$. Notice that in this proof, Step 1 is unnecessary and we start with Step 2.

	 \textbf{Step 2:} 
	Following the arguments similar to Step 2 in the proof of Theorem~\ref{theorem_firstorder}, we conclude that there must exist a normal follower $i\in \mathcal{W}^\mathcal{N}$ either in $\overline{\mathcal{X}}_1(k)$ or $\underline{\mathcal{X}}_1(k)$ which makes an update using a value from the normal nodes outside the set to which it belongs.
	
	According to \eqref{xhk}, if $i\in\overline{\mathcal{X}}_1(k)$ or $i\in\underline{\mathcal{X}}_1(k)$, it implies that $\hat{e}_i[k] = \delta_{N+1}[k]$ or $\hat{e}_i[k] = \delta_{1}[k]$, respectively. Furthermore, if (this can happen when $|\overline{\mathcal{X}}_1(k)| = 1$ if $i\in\overline{\mathcal{X}}_1(k)$ and when $|\underline{\mathcal{X}}_1(k)| = 1$ if $i\in\underline{\mathcal{X}}_1(k)$)
	\begin{align}\label{deltan+1_second}
		\delta_{N+1}[k]-\delta_{N}[k] &> \epsilon_1, \medspace\medspace\medspace \textup{if} \medspace i\in \overline{\mathcal{X}}_1(k), \nonumber \\
		\delta_{2}[k]-\delta_{1}[k] &> \epsilon_1, \medspace\medspace\medspace \textup{if} \medspace i\in \underline{\mathcal{X}}_1(k),
	\end{align}
	and we have shown that there is a node $i$ either in $\overline{\mathcal{X}}_1(k)$ or $\underline{\mathcal{X}}_1(k)$ using a value from the normal nodes outside the set to which it belongs, then one can conclude that the distance of $\hat{e}_i[k]$ from the closest normal $\hat{e}_j[k]$ is larger than $\epsilon_1$. Moreover, since $\mathcal{A}$ is $f$-local, any Byzantine neighbor with consensus error outside the interval $[\delta_{1}[k], \delta_{N+1}[k]]$ is discarded by the MW-MSR algorithm. Therefore, by \eqref{phiik_second} and \eqref{deltan+1_second}, for any possible $a_{ij}[k]\geq0$ obtained in \eqref{phiik_second}, we have
	\begin{align}\label{phiikabsolute_second}
		\phi_i[k]> \epsilon, \medspace\medspace\medspace &\textup{if} \medspace i\in \overline{\mathcal{X}}_1(k), \nonumber \\
		\phi_i[k]< -\epsilon, \medspace\medspace\medspace &\textup{if} \medspace i\in \underline{\mathcal{X}}_1(k)  .
	\end{align}
	Since $|\phi_i[k]| > \epsilon$, we get $\textup{sat}_\epsilon(\phi_i [k] ) = \textup{sgn}(\phi_i[k])$. Then, from \eqref{x1} and \eqref{phiikabsolute_second}, it holds that
	\begin{align*}
		\textup{sat}_\epsilon(\phi_i [k] ) = \textup{sgn}(\phi_i[k])=\textup{sgn}(\hat{e}_i[k]) .
	\end{align*}
	
	In the following, we show that in a finite time, $\hat{e}_i[k]$ is decreasing if $\hat{e}_i[k]>0$, and $\hat{e}_i[k]$ is increasing if $\hat{e}_i[k]<0$.

	(i) If $\hat{e}_i[k]>0$, then by \eqref{phiikabsolute_second}, we have $\phi_i[k]> \epsilon$. We know from \eqref{secondorder} that 
	\begin{align}
		\hat{e}_i[k+1]-\hat{e}_i[k] = (v_i[k] - v_d[k])T ,
	\end{align}
	where $\hat{e}_i[k]=\hat{x}_i[k]-\hat{x}_d[k], i\in \mathcal{V}$. When $v_i[k] < v_d[k]$, we have  $\hat{e}_i[k+1]-\hat{e}_i[k]<0$, i.e., $\hat{e}_i[k]$ is decreasing.
	Then we discuss the following three cases for $v_d[k] \leq v_i[k]$.
	
	Case 1: $0 < v_d[k] \leq v_i[k]$. In this case, by \eqref{input2}, we get
	\begin{align}\label{uik_increasing}
		u_i [k] = -\gamma_i[k] \phi_i [k] - \theta_i \textup{sgn}(\hat{e}_i[k]) -\beta v_i[k]<0 ,
	\end{align}
	indicating that $v_i[k]$ is decreasing until $0<v_i[\hat{k}] < v_d[\hat{k}]$ for some $\hat{k}$. Then we have $\hat{e}_i[k]$ is decreasing for $k\geq \hat{k}$.
	
	Case 2: $ v_d[k]\leq 0 \leq v_i[k]$. In this case, similar to \eqref{uik_increasing}, we get $v_i[k]$ is decreasing until $v_d[\hat{k}]\leq  v_i[\hat{k}] <0$ for some $\hat{k}$.
	
	Case 3: $v_d[k]\leq  v_i[k] <0$. In this case, recall that $\phi_i[k] > \epsilon$.
	Moreover, since $i\in\overline{\mathcal{X}}_1(k)$ or $i\in\underline{\mathcal{X}}_1(k)$, it must hold that $\hat{e}_i[k] > \epsilon$.
	Moreover, by \eqref{beta}, we have
	\begin{align*}
		0< \beta |v_d[k] | \leq \beta \overline{v}_d  \leq \theta_i .
	\end{align*}
	Then we obtain
	\begin{align}
		- \theta_i \textup{sgn}(\hat{e}_i[k]) -\beta v_i[k]< - \theta_i \textup{sgn}(\hat{e}_i[k]) -\beta v_d[k]< 0 .
	\end{align}
	It further follows that
	\begin{align}
		u_i [k] = -\gamma_i[k] \phi_i [k] - \theta_i \textup{sgn}(\hat{e}_i[k]) -\beta v_i[k]<0 ,
	\end{align}
	indicating that $v_i[k]$ is decreasing until $v_i[\hat{k}] < v_d[\hat{k}]$ for some $\hat{k}$. Then we have $\hat{e}_i[k]$ is decreasing for $k\geq \hat{k}$.

	(ii) If $\hat{e}_i[k]<0$, then by \eqref{phiikabsolute_second} we have $\phi_i[k]<- \epsilon$. We know from \eqref{secondorder} that 
	\begin{align}
		\hat{e}_i[k+1]-\hat{e}_i[k] = (v_i[k] - v_d[k])T .
	\end{align}
	When $v_d[k] < v_i[k]$, we have  $\hat{e}_i[k+1]-\hat{e}_i[k]>0$, i.e., $\hat{e}_i[k]$ is increasing.
	Then we discuss the following three cases for $v_i[k] \leq v_d[k]$.
	
	Case 1: $v_i[k] \leq v_d[k] < 0 $. In this case, we have
	\begin{align}\label{uik_decreasing}
		u_i [k] = -\gamma_i[k] \phi_i [k] - \theta_i \textup{sgn}(\hat{e}_i[k]) -\beta v_i[k] >0 ,
	\end{align}
	indicating that $v_i[k]$ is increasing until $v_d[\hat{k}] < v_i[\hat{k}]<0$ for some $\hat{k}$, then we have $\hat{e}_i[k]$ is increasing for $k\geq \hat{k}$.
	
	Case 2: $ v_i[k]\leq 0 \leq v_d[k]$. In this case, similar to \eqref{uik_decreasing}, we have $v_i[k]$ is increasing until $0 <v_i[\hat{k}]\leq  v_d[\hat{k}] $ for some $\hat{k}$.
	
	Case 3: $0<v_i[k]\leq  v_d[k]$. In this case, recall that $\phi_i[k] <- \epsilon$.
	Moreover, since $i\in\overline{\mathcal{X}}_1(k)$ or $i\in\underline{\mathcal{X}}_1(k)$, it must hold that $\hat{e}_i[k] <- \epsilon$.
	Moreover, by \eqref{beta}, we have
	\begin{align*}
		0< \beta |v_d[k] | \leq \beta \overline{v}_d  \leq \theta_i .
	\end{align*}
	Then we obtain
	\begin{align}
		- \theta_i \textup{sgn}(\hat{e}_i[k]) -\beta v_i[k]> - \theta_i \textup{sgn}(\hat{e}_i[k]) -\beta v_d[k]> 0 .
	\end{align}
	It further follows that
	\begin{align}
		u_i [k] = -\gamma_i[k] \phi_i [k] - \theta_i \textup{sgn}(\hat{e}_i[k]) -\beta v_i[k]>0 ,
	\end{align}
	indicating that $v_i[k]$ is increasing until $v_d[\hat{k}] < v_i[\hat{k}]$ for some $\hat{k}$. Then we have $\hat{e}_i[k]$ is increasing for $k\geq \hat{k}$.

	Therefore, we conclude that while $\hat{e}_i[k]>0$ and $\phi_i[k]> \epsilon$, $\hat{e}_i[\hat{k}]$ will always be decreasing after a finite time $\hat{k}$. Moreover, while $\hat{e}_i[k]<0$ and $\phi_i[k]< -\epsilon$, $\hat{e}_i[\hat{k}]$ will always be increasing after a finite time $\hat{k}$.
	Then, in a finite time, we have
	\begin{align*}
		\delta_{N+1}[k]-\delta_{N}[k] &\leq \epsilon_1, \medspace\medspace\medspace \textup{if} \medspace i\in \overline{\mathcal{X}}_1(k),  \\
		\delta_{2}[k]-\delta_{1}[k] &\leq \epsilon_1, \medspace\medspace\medspace \textup{if} \medspace i\in \underline{\mathcal{X}}_1(k)  .
	\end{align*}

	\vspace{-4mm}
	 \textbf{Step 3:} 
	Then we utilize an analysis similar to Step 3 in the proof of Theorem~\ref{theorem_firstorder}. By considering the ultimate error bounds obtained in steps 1 to $h- 1$, for step $h$, in a finite time we should have for each node $i$
	\begin{align*}
		\delta_{N+2-h}[k]-\delta_{N+1-h}[k] &\leq \epsilon_h, \medspace\medspace\medspace \textup{if} \medspace i\in \overline{\mathcal{X}}_h(k),  \\
		\delta_{h+1}[k]-\delta_{h}[k] &\leq \epsilon_h, \medspace\medspace\medspace \textup{if} \medspace i\in \underline{\mathcal{X}}_h(k)  .
	\end{align*}
	Lastly, if for a $1 < h \leq N + 1$, $d \in \overline{\mathcal{X}}_h(k)\cap\underline{\mathcal{X}}_h(k)$; then, by putting the error bounds together, we obtain the consensus error bound in \eqref{errorbound}. The proof is complete. 
\end{proof}

We state a smaller bound for the case of $\mathcal{W}_d=\mathcal{W}$ in the following lemma. It can be derived from Step 2 of the proof of Theorem~\ref{theorem_secondorder}.

\begin{lemma}\label{lemma_secondorder}
	Consider the network $\mathcal{G} = (\mathcal{V},\mathcal{E})$ with the problem settings and assumptions in Theorem~\ref{theorem_secondorder}. Resilient dynamic leader-follower consensus is achieved as in \eqref{reach_consensus} with consensus error bound $\overline{\epsilon}_{d2}=\epsilon$ if all followers are in $\mathcal{W}_d$, i.e., $\mathcal{W}_d=\mathcal{W}$.
\end{lemma}

\begin{remark}\label{smallerbound2}
	Similar to Remark~\ref{smallerbound1}, when $|\mathcal{W}_d\cap \mathcal{N}|<N$, we can have a smaller consensus error bound $\overline{\epsilon}_2<\overline{\epsilon}$ in \eqref{errorbound} for Theorem~\ref{theorem_secondorder}, which is given by
	\begin{align}\label{errorbound_smaller2}
		\overline{\epsilon}_2 = \sum_{h =1}^{N-|\mathcal{W}_d\cap \mathcal{N}|} \epsilon_h + \epsilon   .
	\end{align}
\end{remark}

\subsection{Analysis on Insecure Leader Agents and Discussions on Related Works}\label{discussion_notsecured}

In this section, we study more vulnerable systems with insecure leaders, i.e., a leader may be subject to failures or attacks. There are two important assumptions in this setting.
First, there must be redundancy in the set of leaders; see, e.g., \cite{usevitch2020resilient}.
Second, all normal leaders should broadcast the same reference value at each time; see footnote 1.

In our problem with insecure leaders, we consider the graph $\mathcal{G} = (\mathcal{V},\mathcal{E})$, where $\mathcal{V}$ consists of the set of leader agents $\mathcal{L}$ and the set of follower agents $\mathcal{W}$ with $\mathcal{L}\cup \mathcal{W}=\mathcal{V}$ and $\mathcal{L}\cap \mathcal{W}=\emptyset$. Moreover, the set of normal leaders is denoted by $\mathcal{L}^\mathcal{N}=\mathcal{L}\cap \mathcal{N}$ with $\mathcal{L}^\mathcal{N} \cup \mathcal{W}^\mathcal{N}=\mathcal{N}$. Here, $\mathcal{L}^\mathcal{N}$ may not be $\mathcal{L}$, i.e., there could be adversarial leaders.
We formulate this problem as follows. 
We remark that this problem is the same with the case where secure leaders are unknown to the followers.

\begin{problem}\label{problem2}
	Suppose that all leaders in $\mathcal{L}$ are insecure. 
	Design a distributed control strategy such that the normal agents in $\mathcal{N}$ reach resilient dynamic leader-follower consensus, i.e.,
	for any possible sets and behaviors of the adversaries in $\mathcal{A}$ and any state values of the normal agents in $\mathcal{N}$, \eqref{reach_consensus} is satisfied $\forall d \in \mathcal{L}^\mathcal{N}$.
\end{problem}

To solve Problem~\ref{problem2}, we slightly modify our algorithms for first-order and second-order MASs as follows:

\textbf{Modified Update Rule~\ref{updaterule1}:} Each normal follower agent $i\in \mathcal{W}^\mathcal{N}$ updates its $\phi_i [k]$ according to Algorithm~1, and follows \eqref{system1} with \eqref{input1}.

\textbf{Modified Update Rule~\ref{updaterule2}:} Each normal follower agent $i\in \mathcal{W}^\mathcal{N}$ updates its $\phi_i [k]$ according to Algorithm~1 with inputs $\hat{x}_i[k]$, and follows \eqref{secondorder} with \eqref{input2}.

We note that the nodes in $\mathcal{W}_d$ in Problem~\ref{problem} can be viewed as the insecure leaders in $\mathcal{L}$ for the rest of the follower agents in Problem~\ref{problem2}. The difference is that the nodes in $\mathcal{W}_d$ would have a small consensus error $\epsilon$ to the reference value, while the normal leaders in $\mathcal{L}$ share exactly the correct reference value by our assumption. Therefore, in Problem~\ref{problem2}, we introduce the following graph notion.

\begin{definition}\label{robust_following_leaders}
	The graph $\mathcal{G} = (\mathcal{V},\mathcal{E})$ with the set of leaders $\mathcal{L} \subset \mathcal{V}$ is said to be an $r$-robust following graph with $l$ hops (under the $f$-local model) w.r.t. $\mathcal{L}$ if for any $f$-local set $\mathcal{F}$, the subgraph $\mathcal{G}_{\mathcal{H}}$ with $\mathcal{H}=\mathcal{V}\setminus\mathcal{F}$ satisfies that for every nonempty subset $\mathcal{S}\subseteq \mathcal{H}\setminus \mathcal{L}$, \eqref{zrs} holds.
\end{definition}

We are ready to state the result for solving Problem~\ref{problem2}. Since it can be proved using analyses similar to those in the proofs of Theorems~\ref{theorem_firstorder} and \ref{theorem_secondorder}, we omit its proof.
Note that for the case of insecure leaders, the smaller error bounds in Lemmas~\ref{lemma_firstorder} and \ref{lemma_secondorder} also hold, but the ones in Remarks~\ref{smallerbound1} and \ref{smallerbound2} may not hold since followers in $\mathcal{W}_d$ may not use the leaders' value in their updates.

\begin{proposition}\label{proposition_secure}
	Consider the network $\mathcal{G} = (\mathcal{V},\mathcal{E})$ with the set of insecure leaders $\mathcal{L}$, where all normal followers in $\mathcal{W}^\mathcal{N}$ update their values according to the modified Update Rules~\ref{updaterule1} or \ref{updaterule2}. Under Assumptions~\ref{leader-speed}--\ref{assumption_path} and the $f$-local Byzantine set $\mathcal{A}$, resilient dynamic leader-follower consensus with insecure leaders is achieved with consensus error bound $\overline{\epsilon}$ in \eqref{errorbound} if and only if $\mathcal{G}$ is an $(f+1)$-robust following graph with $l$ hops w.r.t. $\mathcal{L}$.
\end{proposition}

\section{Numerical Examples}\label{sec_example}

%In this section, we conduct simulations of the proposed methods applied in leader-follower networks.

\subsection{Simulations for a First-Order MAS}\label{sec_example1}

We apply Update Rule~\ref{updaterule1} to the leader-follower network in Fig.~\ref{15node} under the $2$-local model. The control input of the leader is given as $u_d[k] = 3 \cos( kT ) + 1.5 \cos(0.5kT)$ with sampling period $T=0.01$ and the external disturbances are bounded with $\overline{\sigma}_i = 1, \forall i \in \mathcal{W}^\mathcal{N}$ (generated using random signals). The objective of normal follower nodes is to track the leader $d$. Accordingly, let $\alpha_i = 1$, $\epsilon = 0.1$, and $\theta_i = 6$.
Moreover, let $x_i[0]\in (0,30), \forall i \in \mathcal{W}^\mathcal{N}$.

\begin{figure}[t]
	\centering
	%\vspace{-10pt}
	\subfigure[\scriptsize{With one-hop communication.}]{
		\includegraphics[width=3.4in,height=1.5in]{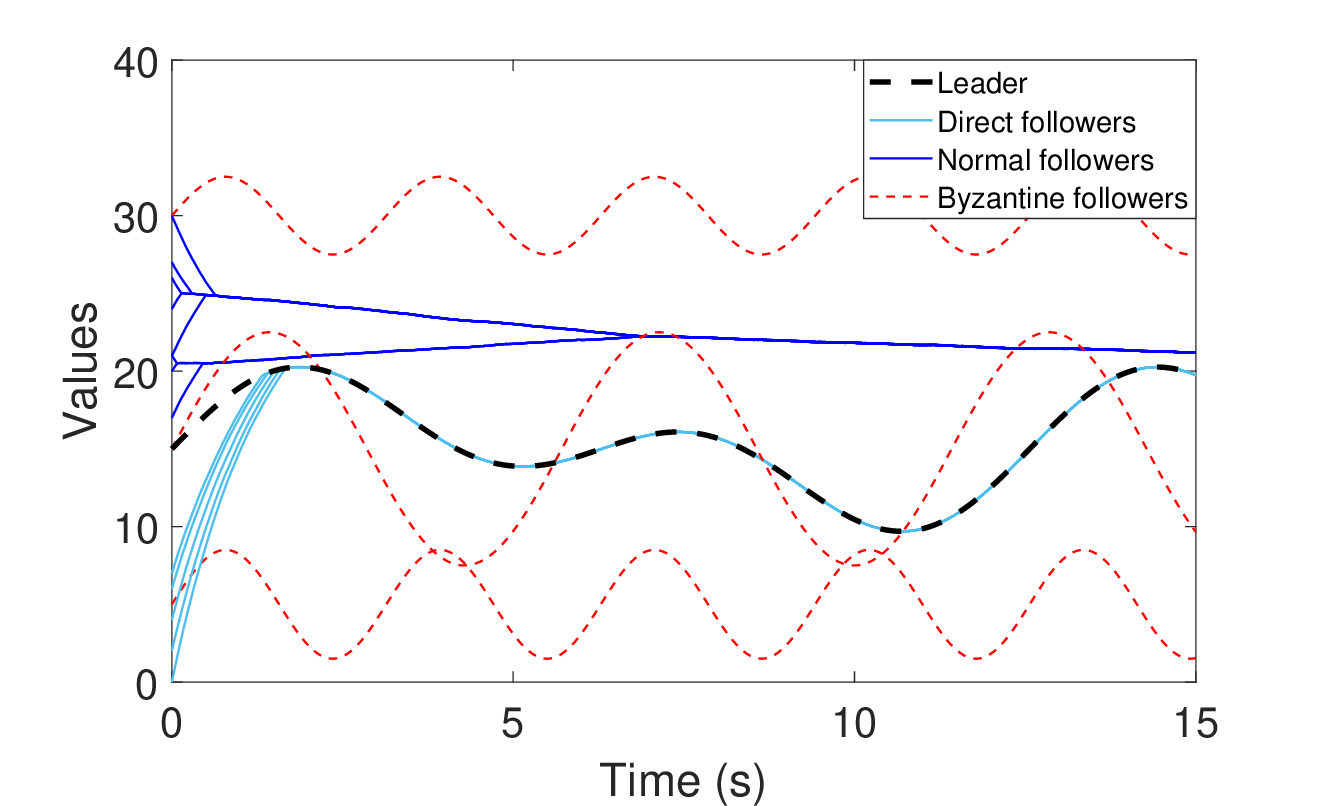}
	}
	\vspace{-9pt}
	
	\subfigure[\scriptsize{With three-hop communication.}]{
		\includegraphics[width=3.4in,height=1.5in]{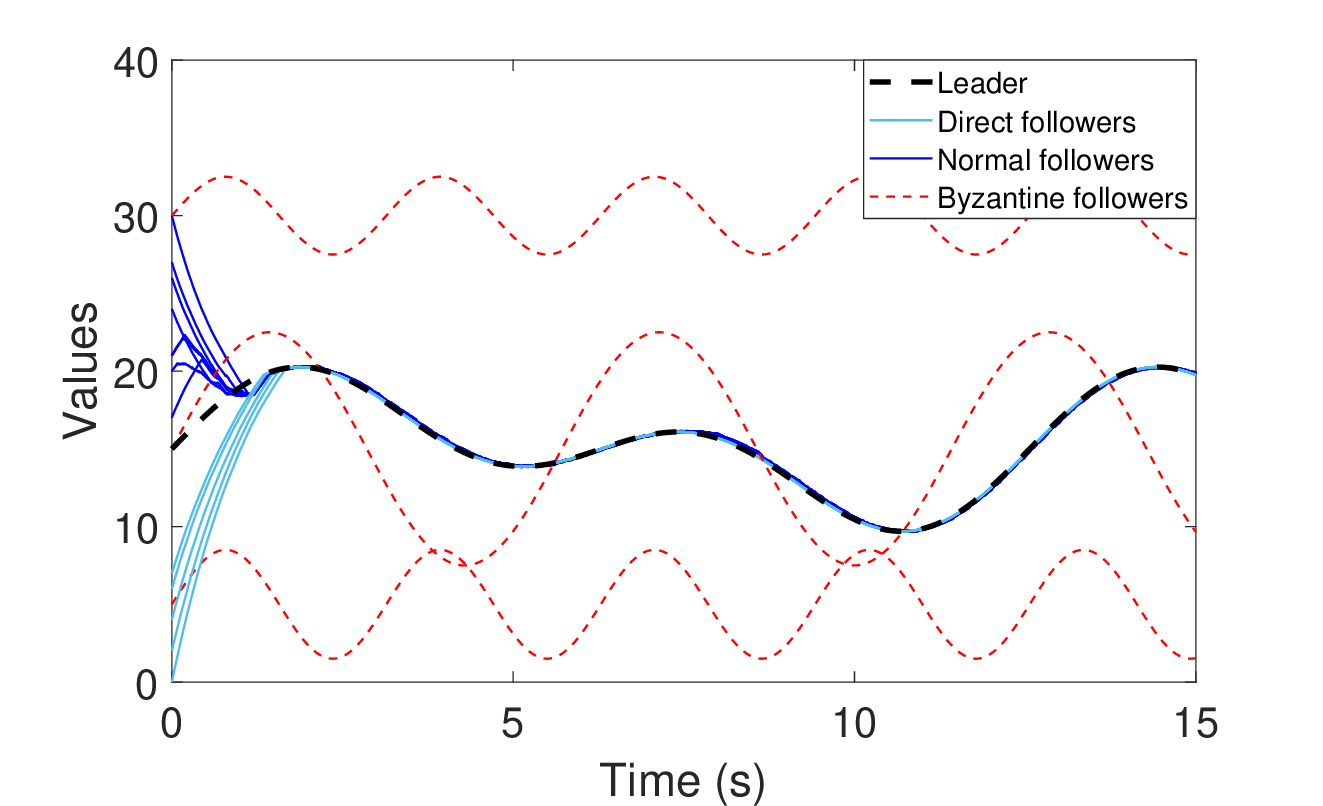}
	}
	\vspace{-9pt}
	\caption{Nodes' values of the network in Fig.~\ref{15node} applying Update Rule~\ref{updaterule1}.}
	\label{15nodevalue}
	%\vspace*{-1.5mm}
\end{figure}

This network is not a $3$-robust following graph with $1$ hop. Hence, it is not robust enough for the one-hop MSR algorithms (\cite{leblanc2013resilient,rezaee2021resiliency}) to succeed under the $2$-local model. However, as mentioned in Example~\ref{discussion15node}, it is a $3$-robust following graph with $3$ hops.
We assume that nodes 8 and 9 are Byzantine. Specifically, node 8 sends an oscillating value around 15 to out-neighbors in node set $\{1, 2\}$ and sends an oscillating value around 5 to out-neighbors in set $\{3, 4, 5\}$. Moreover, node 9 sends an oscillating value around 30 to its neighbors.
First, we present the case where normal followers apply Update Rule~\ref{updaterule1} with one-hop communication, which is similar to the algorithm in \cite{rezaee2021resiliency}. 
The results are given in Fig.~\ref{15nodevalue}(a), and resilient dynamic leader-follower consensus is not achieved.  
Next, we apply Update Rule~\ref{updaterule1} with three-hop communication to this network. We assume that Byzantine nodes manipulate the relayed values in the same way as they manipulate their own values.
Observe in Fig.~\ref{15nodevalue}(b) that resilient dynamic leader-follower consensus is achieved; the maximum consensus error after time $k=500$ (i.e., 5 seconds) is $\overline{\epsilon} =0.37$, verifying the result in Theorem~\ref{theorem_firstorder}.

\begin{figure}[t]
	\centering
	%\vspace{-10pt}
	\subfigure[\scriptsize{With one-hop communication.}]{
		\includegraphics[width=3.4in,height=1.5in]{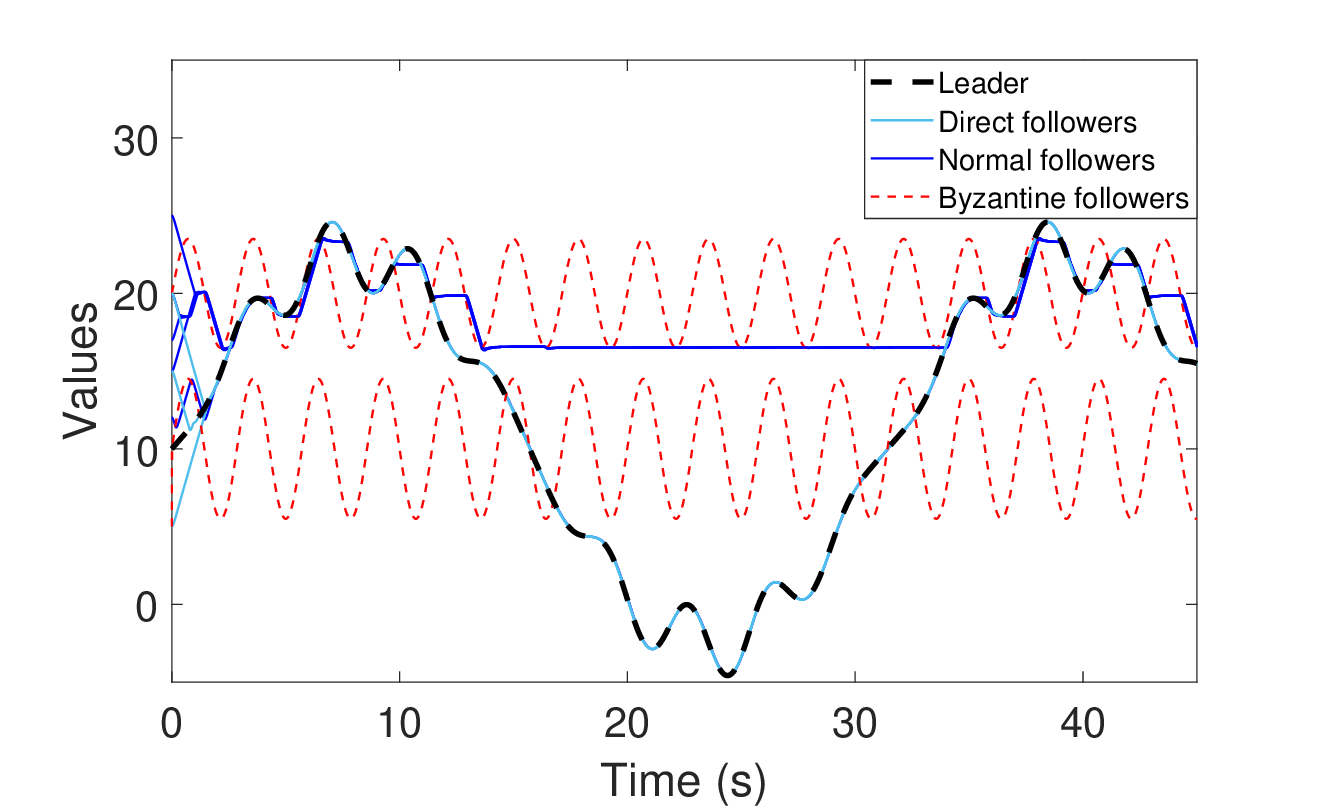}
	}
	\vspace{-9pt}
	
	\subfigure[\scriptsize{With two-hop communication.}]{
		\includegraphics[width=3.4in,height=1.5in]{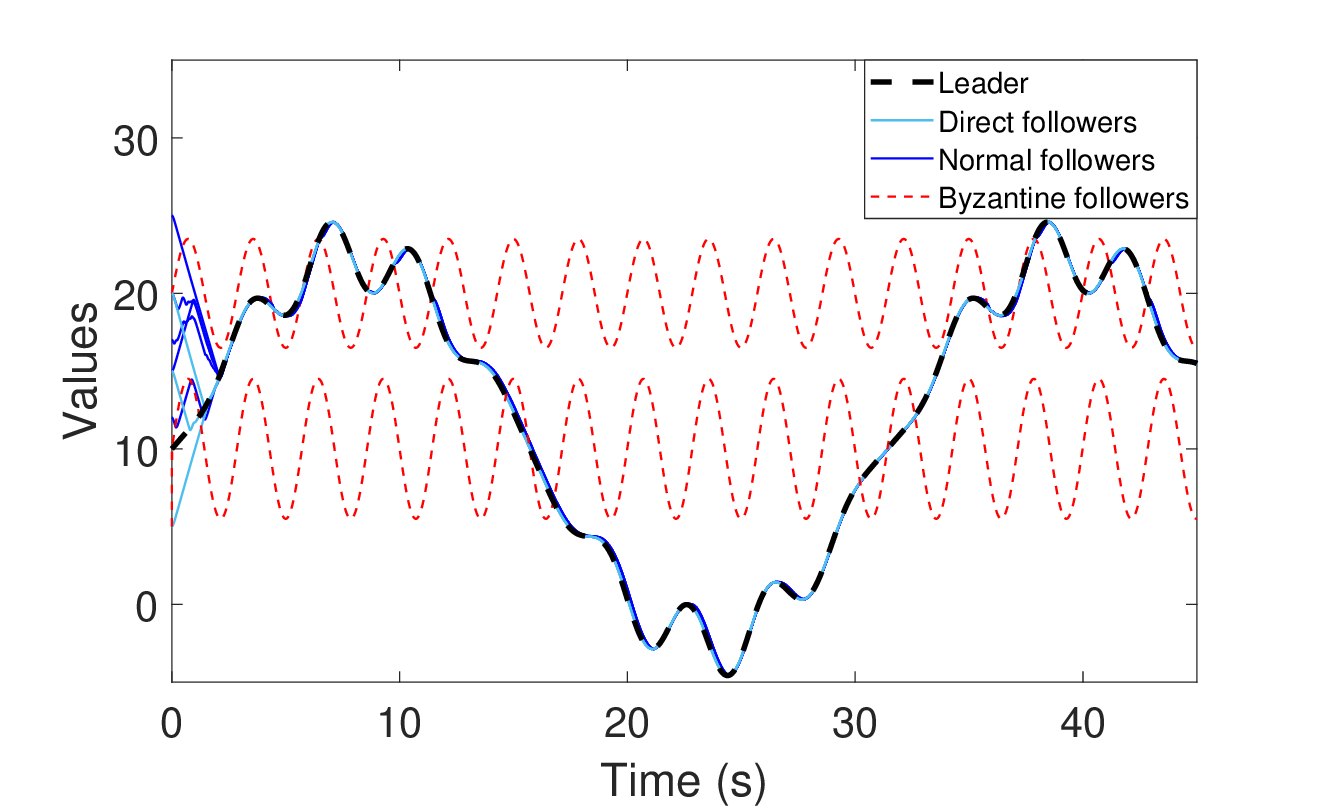}
	}
	\vspace{-9pt}
	\caption{Nodes' values of the network in Fig.~\ref{9node}(a) applying Update Rule~\ref{updaterule2}.}
	\label{9nodevalue}
	%\vspace*{-1.5mm}
\end{figure}

\subsection{Simulations for a Second-Order MAS}\label{sec_example2}

In this part, we present simulations of Update Rule~\ref{updaterule2} in the leader-follower network in Fig.~\ref{9node}(a) under the $1$-local model. Here, all agents in $\mathcal{V}$ possess second-order dynamics as \eqref{secondorder} and \eqref{secondorder-leader}. Scenario~1 presents a one-dimensional tracking problem. Scenario~2 presents a two-dimensional resilient formation control problem.

%ud=2*(-2*0.8*sin(i*0.01)+1.2*1.6*sin(i*0.008)-10*0.025*sin(i*0.001));

\textit{Scenario~1:}
The control input of the leader is given as $u_d[k] = -3.2 \sin( 2kT ) + 1.92 \sin(1.6kT) -0.25 \sin( 0.2kT )$ with $T=0.005$, and the external disturbances are random and bounded with $\overline{\sigma}_i = 1, \forall i \in \mathcal{W}^\mathcal{N}$. The objective of normal followers is to track the leader $d$. In Scenario~1, let $\rho_i=0, \forall i\in \mathcal{V}$.
Moreover, let $\alpha_i = 1$, $\epsilon = 0.1$, $\beta = 20$, and $\theta_i = 100$.
Normal followers have initial values as $\hat{x}_i[0]=x_i[0]\in (0,30), \forall i \in \mathcal{W}^\mathcal{N}$.

As discussed in Example~\ref{discussion9node}, the graph in Fig.~\ref{9node}(a) is not a $2$-robust following graph with $1$ hop, and hence, is not robust enough to tolerate one Byzantine node using the one-hop MSR algorithms (\cite{usevitch2020resilient,rezaee2021resiliency}). 
Here, we assume that Byzantine node 5 sends two different oscillating $\hat{x}$ values to nodes in sets $\{1,3,6\}$ and $\{4\}$.
First, we apply Update Rule~\ref{updaterule2} with one-hop communication to the network in Fig.~\ref{9node}(a). The results in Fig.~\ref{9nodevalue}(a) show that resilient dynamic leader-follower consensus is not reached. For the one-hop algorithm to succeed, the network should have more connections as shown in Fig.~\ref{9node}(b).

\begin{figure}[t]
	\centering
	%\vspace{-10pt}
	\subfigure[\scriptsize{With one-hop communication.}]{
		\includegraphics[width=2.7in]{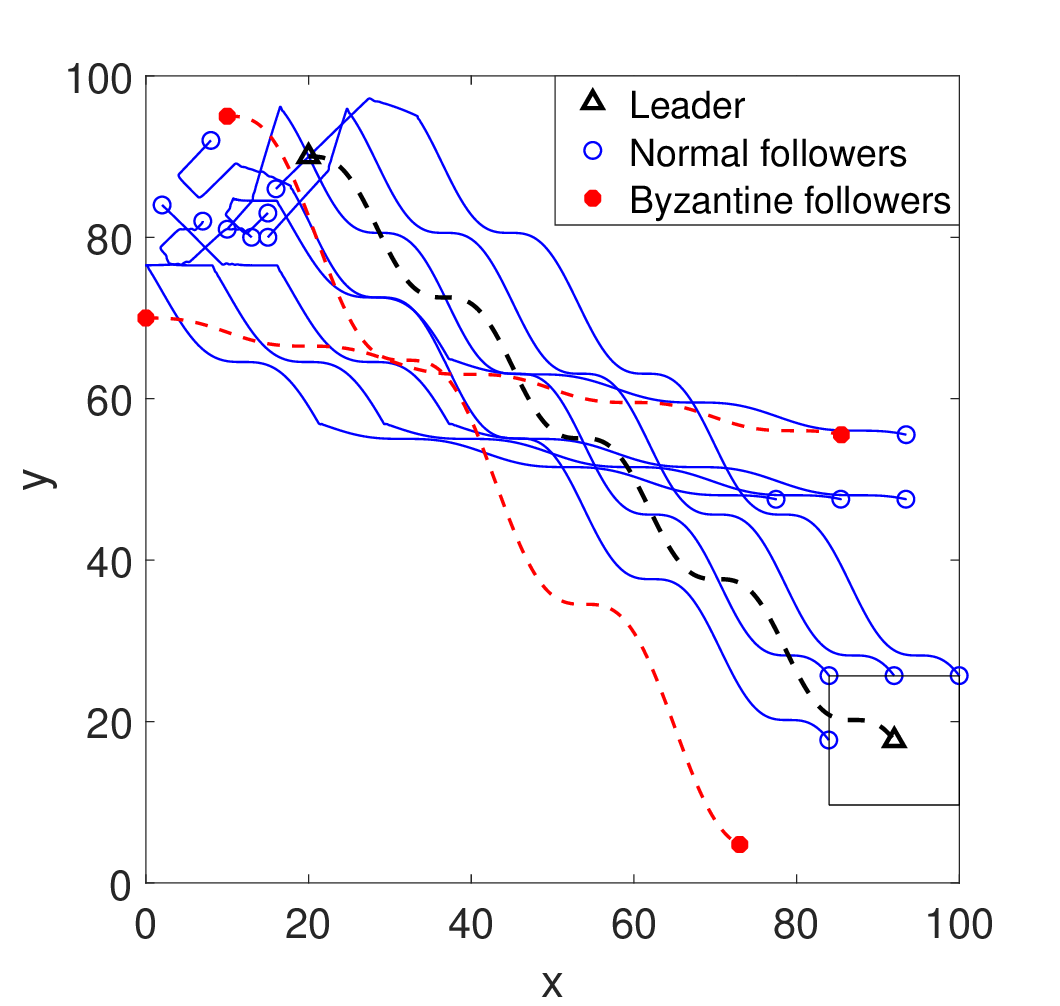}
	}
	\vspace{-9pt}
	
	\subfigure[\scriptsize{With two-hop communication.}]{
		\includegraphics[width=2.7in]{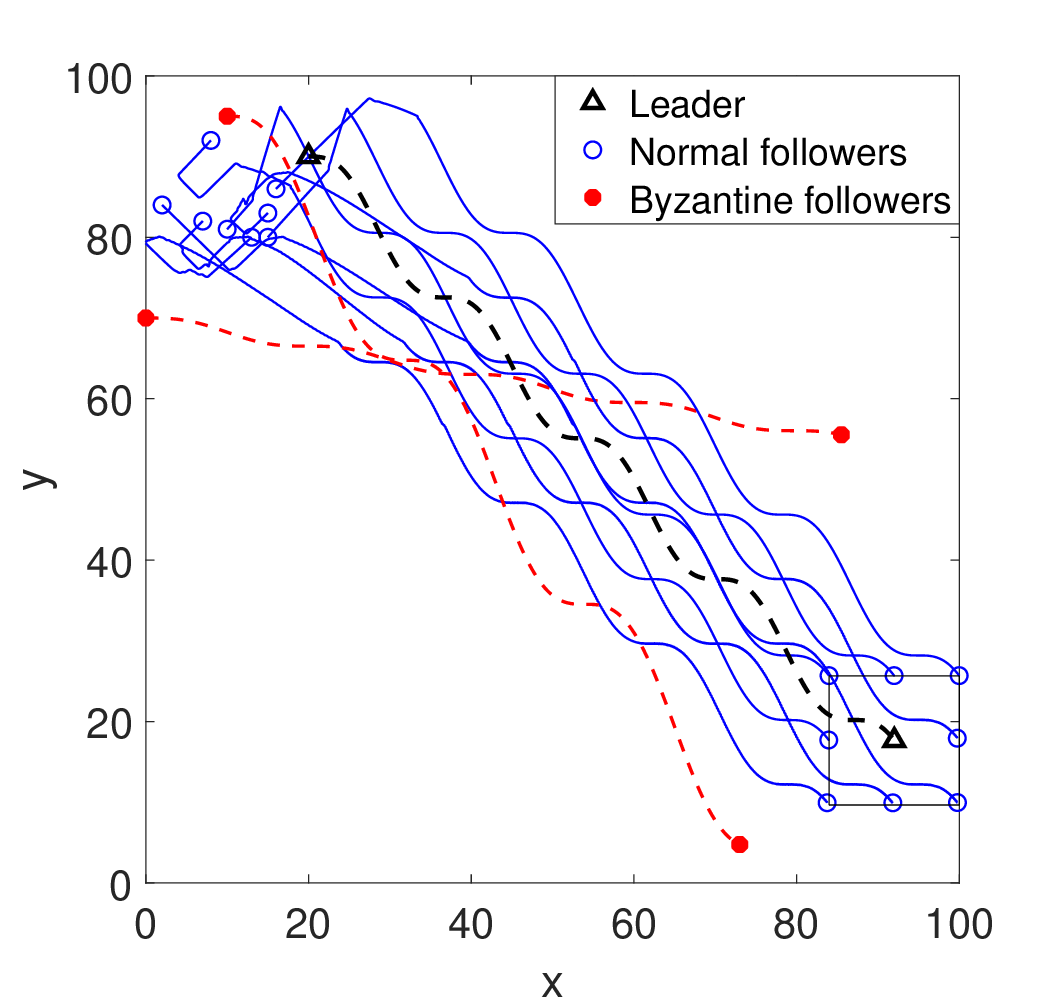}
	}
	\vspace{-9pt}
	\caption{Nodes' trajectories of the network in Fig.~\ref{9node}(a) applying Update Rule~\ref{updaterule2}.}
	\label{9nodeposition2d}
	%\vspace*{-1.5mm}
\end{figure}	

Alternatively, we can increase the network robustness by introducing multi-hop relays with the topology unchanged.
This time, we apply Update Rule~\ref{updaterule2} with two-hop communication to the network in Fig.~\ref{9node}(a), which is a $2$-robust following graph with $2$ hops.
As Theorem~\ref{theorem_secondorder} indicated, it can tolerate one Byzantine node. Here, node 5 relays false values in the same way as it manipulates its own. The results are given in Fig.~\ref{9nodevalue}(b). The maximum consensus error after 2.5 seconds is $\overline{\epsilon} =0.66$. Hence, resilient dynamic leader-follower consensus is achieved.

%ud=-1*sin(i*0.003);

\textit{Scenario~2:}
Our approaches can be extended to decoupled multi-dimensional dynamics of agents. Consider the network in Fig.~\ref{9node}(a), each node $i\in \mathcal{V}$ is associated with two dimensions, i.e., $x_i[k]$ and $y_i[k]$. To be specific, each node $i \in \mathcal{V}$ exchanges $\hat{x}_i[k]$ and $\hat{y}_i[k]$ with neighbors at each time $k$, and we employ Update Rule~\ref{updaterule2} for each node $i\in \mathcal{W}^\mathcal{N}$ on each axis separately. The objective is that the normal followers track the leader in both dimensions to form a desired formation such that the consensus errors on both dimensions are ultimately uniformly bounded in a finite time, despite misbehaviors of the Byzantine node. We employ the same parameters as those in Scenario~1. Let initial values be $x_i[0]\in (0,20), y_i[0]\in (70,100), \forall i \in \mathcal{W}^\mathcal{N}$.

The desired formation is for normal nodes to form a square with the moving leader $d$ located at the square center. The control inputs of the leader on $x$-axis and $y$-axis are given as $u_{d_x}[k] = 0$ and $u_{d_y}[k] = -\sin( 0.6kT )$, respectively.
First, trajectories of normal followers applying the one-hop algorithm are presented in Fig.~\ref{9nodeposition2d}(a). Here, there are two red trajectories of node 5 since it sends two different positions to its neighbors. As a result, the followers are divided into two groups and the desired formation is not achieved. Next, the two-hop algorithm is applied in the same scenario and the results are presented in Fig.~\ref{9nodeposition2d}(b). There, normal agents successfully form the desired formation while tracking the leader. The maximum consensus errors on $x$ and $y$ axes after 5 seconds are $\overline{\epsilon}_x =0.41$ and $\overline{\epsilon}_y =0.49$, respectively. Thus, resilient dynamic leader-follower consensus is achieved on both dimensions. We have verified the efficacy of our methods.

\section{Conclusion}
We have investigated resilient dynamic leader-follower consensus in directed networks.
Our approaches are based on the MW-MSR algorithm from our previous work (\cite{yuan2021resilient}) studying leaderless resilient consensus; besides, the second approach is able to handle agents possessing second-order dynamics.
More importantly, we have characterized tight necessary and sufficient graph conditions for our algorithms. When we employ one-hop communication, our graph conditions are tighter than the ones for the case of insecure leaders (\cite{usevitch2020resilient}) and the case of the secure leader (\cite{rezaee2021resiliency}).
With multi-hop relays, we are able to obtain further relaxed graph requirements for the followers to track the dynamic leader. Moreover, our methods can achieve smaller consensus error bounds than the one in \cite{rezaee2021resiliency}. Possible future directions include to develop approaches for accelerating the MMC calculation and to design new dynamics of followers for reducing consensus errors.

\fontsize{8pt}{8.0pt}\selectfont
%\tiny
%\scriptsize
%\bibliographystyle{dcu}           % Include this if you use bibtex 
%\bibliography{references}           % and a bib file to produce the 
% bibliography (preferred). The
% correct style is generated by
% Elsevier at the time of printing.

\end{document}